\documentclass[12pt]{article}
\usepackage[a4paper]{geometry}
\usepackage{amsmath}
\usepackage{amsthm}
\usepackage{placeins}
\usepackage[title]{appendix}
\usepackage{natbib}
\usepackage[figuresright]{rotating}
\usepackage{enumitem}
\usepackage{algorithm}
\usepackage{algpseudocode}
\usepackage[colorlinks,citecolor=blue,urlcolor=blue]{hyperref}
\usepackage{subfigure}
\usepackage{xcolor}
\usepackage{tikz}
\usetikzlibrary{arrows,automata,er,positioning,calc,arrows}
\usetikzlibrary{shapes,decorations,arrows,calc,arrows.meta,fit,positioning}
\tikzset{
	-Latex,auto,node distance =1 cm and 1 cm,semithick,
	state/.style ={ellipse, draw, minimum width = 0.7 cm},
	point/.style = {circle, draw, inner sep=0.04cm,fill,node contents={}},
	bidirected/.style={Latex-Latex,dashed},
	el/.style = {inner sep=2pt, align=left, sloped}
}
\usepackage{upgreek}
\usepackage{cancel}

\usepackage{amssymb}
\usepackage{mathtools}
\usepackage{xr}

\usepackage[font=normalsize]{caption}
\usepackage{bm}% bold math symbol support (DPC/FMi)
\newcommand{\vectr}[1]{\ensuremath{\bm{\mathrm{#1}}}}% vector [bm]
\newcommand{\matrx}[1]{\ensuremath{\bm{\mathrm{#1}}}}% matrix [bm]
\DeclareMathOperator{\E}{\mathbb{E}}
\DeclareMathOperator{\V}{\mathbb{V}ar}
\DeclareMathOperator{\cov}{\mathbb{C}ov}

\newtheorem{theorem}{Theorem}
\newtheorem{lemma}{Lemma}

\newtheorem{corollary}{Corollary}

\begin{document}
	
	\def\spacingset#1{\renewcommand{\baselinestretch}%
		{#1}\small\normalsize} \spacingset{1}

	\title{\bf Unlocking Retrospective Prevalent Information in EHRs - a Pairwise Pseudolikelihood Approach}
	\author{Nir Keret and Malka Gorfine \hspace{.2cm}\\
		Department of Statistics and Operations Research\\
		Tel Aviv University, Israel}
	\date{}
	\maketitle

	\begin{abstract}
		Typically, electronic health record data are not collected towards a specific research question. Instead, they comprise numerous observations recruited at different ages, whose medical, environmental and oftentimes also genetic data are being collected. Some phenotypes, such as disease-onset ages, may be reported retrospectively if the event preceded recruitment, and such observations are termed ``prevalent". The standard method to accommodate this ``delayed entry" conditions on the entire history up to recruitment, hence the retrospective prevalent failure times are conditioned upon and cannot participate in estimating the disease-onset age distribution. An alternative approach conditions just on survival up to recruitment age, plus the recruitment age itself. This approach allows incorporating the prevalent information but brings about numerical and computational difficulties. In this work we develop consistent estimators of the coefficients in a regression model for the age-at-onset, while utilizing the prevalent data. Asymptotic results are provided, and simulations are conducted to showcase the substantial efficiency gain that may be obtained by the proposed approach. In particular, the method is highly useful in leveraging large-scale repositories for replicability analysis of genetic variants. Indeed, analysis of urinary bladder cancer data reveals that the proposed approach yields about twice as many replicated discoveries compared to the popular approach. 
		
	\end{abstract}
	
	{\bf keywords:} EHR; Left truncation; Pairwise Pseudolikelihood; Prevalent; Replicability; Survival analysis.
	
	\spacingset{2} % DON'T change the spacing!
	
	\section{Introduction}
	Biobanks and Electronic Health Records (EHRs) offer extensive genetic and environmental data. Although not disease-specific, they encompass high-quality information for diverse health studies. Initiatives like the UK Biobank (UKB), China Kadoorie Biobank, Biobank Sweden, FinnGen and many others underscore their expanding popularity and utility. However, fully unlocking their potential necessitates addressing inherent limitations and biases in this type of data.
	
	Biobanks and EHRs often involve delayed-entry scenarios where participants join follow-up at an age (recruitment time) later than the time axis origin, and are then prospectively monitored until death, dropout, or study conclusion. This setup introduces left truncation, as participants must survive long enough to be recruited. The ``prevalent" observations have been diagnosed with the disease of interest before recruitment, reporting the age-at-onset retrospectively. In contrast, ``incidents" are recruited healthy and their onset is observed during follow-up, whereas ``censored" cases do not experience the event by the time of analysis. It is well known that accounting for left truncation is crucial to avoid bias, and care should be taken when integrating prevalent and incident data.
	
	The UKB provides data on approximately 500,000 UK individuals. Notably, participants aged 40 to 69 were enrolled between 2006 and 2010, introducing delayed entry. In relation to urinary bladder cancer (UBC), the subject of Section 4, there are around 880 incident and 590 prevalent cases, so that the latter constitutes about 40\% of all observed events.
	
	Most time-to-event EHR data analyses do not use prevalent cases \citep{pang2018adiposity,gorfine2021marginalized,abhari2022external,keret2023analyzing} due to two key reasons. Firstly, the primary interest is in associating risk factors to the studied disease. However, baseline measurements from prevalent cases, collected post-diagnosis, are susceptible to recall bias, especially for past habits like smoking, drinking, diet, and physical activity. This work leverages the prevalent cases in an important and popular application of EHR data, ensuring that data are accurately collected. Secondly, computational challenges involving numerical instability and long running times have so far hindered utilization of prevalent cases, as will be elaborated later. Our novel approach successfully circumvents this challenge, enabling seamless integration of prevalent cases.
	
	Detecting novel statistical associations between a rare disease and genetic variants requires an ample number of observed events, as the significance threshold in genome-wide association studies (GWAS) is commonly set at $5\times10^{-8}$. These studies are often conducted using multi-center case-control cohorts for increasing the observed event counts \citep{zhang2014large,huyghe2017contribution}. As most genetic studies are exploratory, replication analyses are crucial due to false-positives \citep{kraft2009replication}. Indeed, biobanks are often leveraged as independent cohorts for external replication analyses, with the aim of verifying or challenging prior research findings. 
	
	Section 4 conducts a replication analysis on single nucleotide polymorphisms (SNPs) previously associated with UBC, using UKB data as an independent cohort.  
	Employing a Cox model \citep{cox1972regression}, we observe higher statistical power with the proposed approach compared to the standard partial-likelihood (PL) estimator, adjusted for left truncation and excluding prevalent observations. Of 31 tested SNPs, 11 were significantly associated with increased UBC risk using the proposed approach, compared to six SNPs detected by the standard PL estimator. The Benjamini-Hochberg (BH) \citep{benjamini1995controlling} procedure for multiple-testing was used, with significance threshold set at 0.05.

	\subsection{Related Work}
	
	We assume that conditionally on the covariates, recruitment times are independent of disease-onset times, and quasi-independent \citep{tsai1990testing} of death and censoring times, and this assumption underlies the subsequent discussion. Quasi-independence, intuitively, can be thought of as independence in the observed region, and is therefore weaker than full independence. When the observed events are all incident, a widely applicable method for accommodating left truncation is the ``risk-set adjustment" \citep[pg. 313]{klein2003survival}. At each time point, only participants who have already entered the study and remained uncensored and event-free are regarded at risk. This is the standard left-truncation method for the PL, Kaplan-Meier \citep{kaplan1958nonparametric} and Nelson-Aalen \citep{nelson1972theory,aalen1978nonparametric} estimators, to name a few.
	
	As elaborated in Section 2, when both prevalent and incident cases are present, the disease times are typically  embedded within the ``illness-death model" -- a three-state stochastic model with initial, transient and absorbing states (``healthy", ``diseased" and ``death", respectively), and three possible transitions: ``healthy$\rightarrow$diseased" ($1\rightarrow2$), ``healthy$\rightarrow$dead" ($1\rightarrow3$) and ``diseased$\rightarrow$dead" ($2\rightarrow3$), as depicted in Figure \ref{Fig:ill-death}. Two main approaches for combining the prevalent and incident cases are inverse probability weighting (IPW) and conditional likelihood.
	
	\cite{copas2001incorporating} presented a pseudo-(partial-) likelihood IPW method where each observation is weighted inversely to its inclusion probability. \cite{chang2006nonparametric} and \cite{vakulenko2017nonparametric} proposed nonparametric IPW estimators for the joint distribution of disease and death times, but did not include covariates. \cite{li2011quantile,li2014varying} address semi-competing risks while including the prevalent cases, however these are not applicable to the illness-death model, as the death-time distribution is assumed unaltered by disease occurrence. 
	
	Importantly, these methods are subject to a ``positivity" condition, namely, that each observation in the target population has a positive recruitment probability. In most biobanks this condition is violated. In particular, in the UKB, those who died before age 40 have zero recruitment probability. Additionally, the distribution of recruitment time should be estimated, which we would rather avoid. Hence, the IPW approach will not be further considered in this work. 
	
	As to conditional likelihood, one approach accommodates delayed entry by conditioning on survival until recruitment age. As explained by \cite{vakulenko2016comparing}, unless a parametric model is specified for recruitment ages, they can be conditioned upon without loss of efficiency. While parametric modeling might increase efficiency when correctly specified, it is established that misspecification can induce severe bias, hence we find such an approach unattractive.
	
	A second option is to condition on both survival until recruitment age and the actual recruitment age, eliminating its randomness and the need for distribution specification. Nonetheless, the likelihood in this approach involves all three transitions of the illness-death model, and necessitates numerical integration for each and every observation during the iterative optimization routine, as shown in Section 2. \cite{vakulenko2016comparing} demonstrate that convergence issues and instability can emerge, especially as the sample size increases, even within fully-parametric models for all transitions. Adopting a semi-parametric model is anticipated to worsen instability because the integrand becomes even more complex.
	
	The third, widely-used and standard option, conditions on all available information up to recruitment. The age-at-onset of prevalent observations is conditioned upon, hence they do not contribute to the likelihood of transition $1\rightarrow2$. The advantage of this option is that under standard assumptions the likelihood of the entire illness-death model factorizes into separate components corresponding to the three transitions, so that each can be analyzed independently using marginal models. Since the remaining observed events in transition $1\rightarrow2$ are all incident, the risk-set adjustment can be applied \citep[Section S10 in the supplementary material]{gorfine2021marginalized}. However, omitting prevalent observations can substantially reduce efficiency compared to the first two options, as evidenced by \cite{saarela2009joint} and \cite{vakulenko2016comparing}. This is also demonstrated in Sections 3 and 4 through simulations and real data analysis. 
	
	\subsection{Our Contribution}
	
	The focus of this work is transition $1\rightarrow2$, as it is particularly susceptible to efficiency loss with  the standard PL-based estimation that excludes the prevalent data. We build on the pairwise pseudolikelihood idea of \cite{liang2000regression}, and develop an alternative procedure acting as a proxy for the conditional likelihood given survival until recruitment, and recruitment age. By circumventing the computationally-problematic numerical integration, we propose a stable and reliable estimation procedure.
	
	The proposed method is versatile and can be applied to various parametric or semi-parametric regression models for survival data. However, we present the estimation procedure, data analysis, simulations and asymptotic properties specifically for the Cox regression model due to its widespread popularity. Proofs establishing the consistency and asymptotic normality are provided, as well as a variance estimation procedure.
Importantly, the simulations demonstrate high robustness against model misspecification of the other two transitions, and of censoring, which should also be estimated when assumed random. Lastly, our approach employs all observation pairs, which can be computationally intensive. To address this, we have incorporated a subsampling technique, considerably cutting down computation time without sacrificing efficiency. 
	
\section{Methodology}
	Let $T_1$ and $T_2$ be the ages at disease diagnosis and death, respectively, and $\vectr{Z}$ is a vector of time-independent covariates of size $p$. Since the disease cannot occur after death, similarly to \cite{xu2010statistical}, the probability distribution
	of ($T_1, T_2$) given $\vectr{Z}$ is assumed to be absolutely continuous in the upper wedge $t_2 \ge t_1$. Namely, the joint density of ($T_1, T_2$) given \vectr{Z}, denoted by $f_{T_1,T_2|\vectr{Z}}(t_1, t_2|\vectr{z})$ is defined for $t_2 \ge t_1 \ge 0$, so $$\int_{0}^\infty \int_t^\infty f_{T_1,T_2}(t, v|\vectr{Z})dvdt=\Pr(T_1<\infty|\vectr{Z})\le 1\, ,$$
	and let $T_1 = \infty$ for those who died disease-free.  
	Based on Figure \ref{Fig:ill-death}, let the instantaneous hazard functions of transitioning from state 1 to either state $k$ = 2 or 3, given $\vectr{Z}$, be
	\begin{eqnarray}
		h_{1k}(t|\vectr{Z}) &=& \lim_{\epsilon \searrow 0} \frac{1}{\epsilon} 
		\Pr(t \leq T_{k-1} < t +\epsilon| T_{1} \geq t,T_{2} \geq t,\vectr{Z})
		\,\, , \,\,\, t>0 \,\,, \,\, \, \, k=2,3  \nonumber 
	\end{eqnarray}
	and the cumulative hazard functions are 
	$ H_{1k}(t|\vectr{Z}) = \int_{0}^th_{1k}(s|\vectr{Z})ds, \, \, \, k=2,3 \, .$
	Likewise, the corresponding hazard functions for leaving state 2 given $\vectr{Z}$ and $T_1=t_1$, are
	$$h_{23}(t|\vectr{Z},t_1) = \lim_{\epsilon \searrow 0} \frac{1}{\epsilon} 
	\Pr(t \leq T_{2} < t +\epsilon| T_{1} = t_1,T_{2} \geq t,\vectr{Z})
	\,\, , \,\,\, t>t_1>0 \, , $$
	and,
	$H_{23}(t|\vectr{Z},t_1) = \int_{t_1}^th_{23}(s|\vectr{Z},t_1)ds$. These hazard functions may include infinite-dimensional parameters. Note that although the same covariate vector $\vectr{Z}$ is used in all hazard functions, any selected regression model permits us to assign a coefficient of zero to any specific variable. This presentation style is a notational convenience and does not restrict us from employing distinct covariates across models.
	
	 Now, assume we are given a sample of $n$ independent and identically-distributed observations, such that the recruitment (delayed entry) and observed ages of observation $i$ are $R_i$, and $V_i=\min(T_{1i},T_{2i},C_i)$, respectively, where $C_i$ is its age at right-censoring, and censoring is assumed to occur only after recruitment \citep{qian2014assumptions}. Let $\Delta_{li} = I(V_i=T_{li})$, $l=1,2$, where $I(\cdot)$ is the indicator function, so $\Delta_{1i}=1$ indicates observing the disease onset of observation $i$, and $\Delta_{2i}=1$ indicates observing its disease-free death. When $\Delta_{1i}=\Delta_{2i}=0$, observation $i$ is censored. Denote $\vectr{Z}_i$ as the vector of covariates associated with observation $i$, so overall its observed information is $\{V_i,\Delta_{1i},\Delta_{2i},R_i,\vectr{Z}_i\}$. We assume that conditionally on the covariates, censoring is independent of the failure times and quasi-independent of recruitment time. It is also assumed that the censoring and other three transitions do not share common parameters, but may share common covariates.
	
	Denote $\vectr{O}_i=(V_i,\Delta_{1i},\Delta_{2i})^T$ as the outcome associated with observation $i$. As outlined in Section 1, estimation employs one of three likelihood functions, corresponding to the distribution of $\vectr{O}$ conditional on varying information subsets: 
	\textbf{I.} $\{\vectr{Z}, T_2 > R\}$.
	\textbf{II.} $\{\vectr{Z},R,T_2 > R\}$. 
	\textbf{III.} $(\vectr{Z},R)$ and the entire observed data up to age $R$. 
	
	The conditional likelihood of option I requires specification of the distribution of $R$, which we prefer to avoid for potential misspecification bias \citep{vakulenko2016comparing}. The conditional likelihood of option III for transition $1\rightarrow2$ can be expressed as
	\begin{eqnarray}
		L^{\mbox{III}} &\propto& \prod_{i:R_i<V_{i}}\left\{\frac{h_{12}(V_i|\vectr{Z}_i)^{\Delta_{1i}}\exp\{-H_{12}(V_i|\vectr{Z}_i)\}}{\exp\{-H_{12}(R_i|\vectr{Z}_i)\}} \right\} \, ,\nonumber
	\end{eqnarray}
	which is convenient as it involves only parameters of this transition. However, the prevalent cases do not participate in $L^{\mbox{III}}$, and  instead one can use the likelihood of option II, which uses all observations and involves the entire illness-death process. Namely,
	\begin{eqnarray}
		L^{\mbox{II}} &=&
		\prod_{i=1}^n\frac{f(V_i,\Delta_{1i},\Delta_{2i},T_{2i}>R_i|\vectr{Z}_i,R_i)}{\Pr(T_{2i}>R_i|\vectr{Z}_i,R_i)}  \label{Eq:Like2}\\
		&\propto& \prod_{i=1}^n\frac{h_{12}(V_{i}|\vectr{Z}_i)^{\Delta_{1i}}h_{13}(V_{i}|\vectr{Z}_i)^{\Delta_{2i}}\exp\{-H_{1\cdot}(V_i|\vectr{Z}_i)-\Delta_{1i}I(R_i>V_i)H_{23}(R_i|V_i,\vectr{Z}_i)\}}{\exp\{-H_{1\cdot}(R_i|\vectr{Z}_i)\}+ \int_{0}^{R_i}h_{12}(s|\vectr{Z}_i)\exp\{-H_{1\cdot}(s|\vectr{Z}_i)-H_{23}(R_i|s,\vectr{Z}_i) \}ds}  \, , \nonumber
	\end{eqnarray}
	where $H_{1\cdot}(\cdot|\vectr{Z}) = H_{12}(\cdot|\vectr{Z}) + H_{13}(\cdot|\vectr{Z})$. The denominator in Eq.(\ref{Eq:Like2}) is the probability sum of survival until recruitment with and without the disease. Numerical integration is required for each observation within the optimization routine, which is likely to induce convergence and instability problems \citep{vakulenko2016comparing}, especially when adopting a semi-parametric approach. Below, we present an alternative estimation procedure, acting as a computationally-friendly proxy for likelihood $L^{\mbox{II}}$ that leverages the prevalent information.
	
	\begin{figure}
		\centering
		\includegraphics[width=85mm]{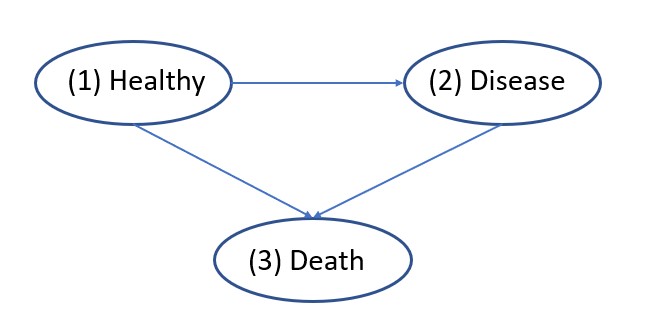}
		\caption{\label{Fig:ill-death} The illness-death model.}
	\end{figure}
	
	\subsection{The Proposed Approach}
	
	\cite{kalbfleisch1978likelihood} elegantly linked between regression permutation tests to score tests based on conditional likelihoods given the order statistic. In some settings this conditional likelihood may help avoiding nuisance parameter estimation. Inspired by this approach, \cite{liang2000regression}, introduced the pairwise pseudolikelihood as a substitute for the computationally-intensive full conditional likelihood, which requires exhaustive enumeration of all $n!$ permutations.
	
	In the pairwise pseudolikelihood, each observation pair contributes their joint distribution conditional on their order statistic. By extending this idea to likelihood $L^{\mbox{II}}$, we can eliminate the denominator in Eq.(\ref{Eq:Like2}), which requires the troublesome numerical integration. Let $\left(\vectr{O}_{(1)},\vectr{O}_{(2)}\right)_{ij}$ be a random permutation of $(\vectr{O}_{i},\vectr{O}_{j})$, it then follows that
	\begin{equation} \label{Eq:pairwise} 
		L^{pair} = \prod_{i<j}L^{pair}_{ij} \, ,
	\end{equation} 
	where the contribution of each pair is
	\begin{eqnarray}
		L^{pair}_{ij} &=& f\left\{\vectr{O}_i,\vectr{O}_j|R_i,R_j,\vectr{Z}_i,\vectr{Z}_j,R_i<T_{2i},R_j<T_{2j},(\vectr{O}_{(1)},\vectr{O}_{(2)})_{ij}\right\} \nonumber \\
		&=& \frac{f\left(\vectr{O}_i,\vectr{O}_j|R_i,R_j,\vectr{Z}_i,\vectr{Z}_j,R_i<T_{2i},R_j<T_{2j}\right)}{f\left\{(\vectr{O}_{(1)},\vectr{O}_{(2)})_{ij}|R_i,R_j,\vectr{Z}_i,\vectr{Z}_j,R_i<T_{2i},R_j<T_{2j}\right\}} \, .  \label{Eq:LOnePair}
	\end{eqnarray}
	Due to independence, the numerator is 
	$f\left(\vectr{O}_i|R_i,\vectr{Z}_i,R_i<T_{2i}\right)f\left(\vectr{O}_j|R_j,\vectr{Z}_j,R_j<T_{2j}\right) \, , $
	and the denominator is $$f\left(\vectr{O}_i|R_i,\vectr{Z}_i,R_i<T_{2i}\right)f\left(\vectr{O}_j|R_j,\vectr{Z}_j,R_j<T_{2j}\right) + f\left(\vectr{O}_j|R_i,\vectr{Z}_i,R_i<T_{2i}\right) f\left(\vectr{O}_i|R_j,\vectr{Z}_j,R_j<T_{2j}\right) \, , $$
	where $f\left(\vectr{O}_j|R_i,\vectr{Z}_i,R_i<T_{2i}\right)$, for instance, is the conditional distribution function of a ``quasi-observation" with outcome $\vectr{O}_j$, recruitment age $R_i$ and covariates $\vectr{Z}_i$. Plugging these expressions back in Eq.(\ref{Eq:LOnePair}), we get
	\begin{eqnarray}
		L^{pair}_{ij}(\vectr{\theta}) &=&
		\frac{\frac{f\left(\vectr{O}_i,R_i<T_{2i}|R_i,\vectr{Z}_i\right)}{\Pr(R_i<T_{2i}|R_i,\vectr{Z}_i)}\frac{f\left(\vectr{O}_j,R_j<T_{2j}|R_j,\vectr{Z}_j\right)}{\Pr(R_j<T_{2j}|R_j,\vectr{Z}_j)}}{\frac{f\left(\vectr{O}_i,R_i<T_{2i}|R_i,\vectr{Z}_i\right)}{\Pr(R_i<T_{2i}|R_i,\vectr{Z}_i)}\frac{f\left(\vectr{O}_j,R_j<T_{2j}|R_j,\vectr{Z}_j\right)}{\Pr(R_j<T_{2j}|R_j,\vectr{Z}_j)} + \frac{f\left(\vectr{O}_j,R_i<T_{2i}|R_i,\vectr{Z}_i\right)}{\Pr(R_i<T_{2i}|R_i,\vectr{Z}_i)}\frac{f\left(\vectr{O}_i,R_j<T_{2j}|R_j,\vectr{Z}_j\right)}{\Pr(R_j<T_{2j}|R_j,\vectr{Z}_j)}} \nonumber \\ 
		&=& \frac{1}{1 + \frac{m_{ji}m_{ij}}{m_{ii}m_{jj}}} \, . \nonumber
	\end{eqnarray}
	
	The terms $\Pr(R_i<T_{2i}|R_i,\vectr{Z}_i)$ and $\Pr(R_j<T_{2j}|R_j,\vectr{Z}_j)$ cancel out, so that
	\begin{eqnarray}
		m_{ji} &=&  h_{12}(V_j|\vectr{Z}_i)^{\Delta_{1j}}h_{13}(V_j|\vectr{Z}_i)^{\Delta_{2j}}h_C(V_j|\vectr{Z}_i)^{1-\Delta_{1j} - \Delta_{2j}}\exp\left\{-H_{1\cdot}(V_j|\vectr{Z}_i) \right. \nonumber \\ & & \left. -H_{23}(R_i|\vectr{Z}_i,V_j)I(V_j<R_i) - H_C(V_j|\vectr{Z}_i)I(V_j>R_i) \right\} I(V_j > R_i)^{1-\Delta_{1j}} \label{Eq:mij} \, ,
	\end{eqnarray}
	where $h_C$ and $H_C$ are the instantaneous and cumulative hazard functions of censoring. In the case of non-random censoring mechanisms like Type 1 censoring \citep[chapter 3.2]{klein2003survival}, these terms do no appear in Eq.(\ref{Eq:mij}), and need not be estimated. In Section 3 we show that our estimation procedure is robust against model misspecifiction for censoring.
	
	Every observation satisfying $V<R$ is prevalent, indicating it has been diagnosed with the disease. However, it may be the case that upon swapping the outcomes within a pair, we end up with a ``quasi-observation" having $V<R$, but that either died or was censored before disease onset. This creates an invalid pair and its corresponding pairwise pseudolikelihood is equal 1, thanks to the last indicator $I(V_j > R_i)^{1-\Delta_{1j}}$ in $m_{ji}$ and the corresponding indicator in $m_{ij}$.
	
	Additionally, under the assumption that the recruitment distribution is independent of the covariates, it would be possible to use the pairwise pseudolikelihood also as a proxy for likelihood I, and avoid estimating the recruitment distribution, as proposed by \cite{huang2013semiparametric} and \cite{wu2018pairwise}. However, we believe that this independence assumption is unrealistic and prefer to avoid it. 
	
	Finally, to enhance efficiency, one could explore a triplet-wise pseudolikelihood (or higher-order tuples), or consider drawing a subset of the total $n!$ permutations in the original conditional likelihood presented in \cite{kalbfleisch1978likelihood}.
	However, \cite{liang2000regression} report that the latter option, of randomly drawing permutations, attains negligible improvement upon the pairwise pseduolikelihood. Furthermore, in our setting, as more observations are involved in a tuple/permutation, the more likely it becomes disqualified, as invalid ``quasi-observations" are likely to appear. Therefore, we adhere to the pairwise pseudolikelihood. 
	
	So far, the derivations were given in general form in terms of the distributions of $T_1$ and $T_2$. In what follows we focus specifically on the Cox model.
	
	\subsection{Cox Model - The Proposed Pairwise Pseudolikelihood}
	Cox models are postulated for the three transitions of Fig.(\ref{Fig:ill-death}), as well as for the censoring distribution. Namely, for $k \in \{12,13,C\}$ it is assumed that $h_k(t|\vectr{Z})=h_{0k}(t)e^{\vectr{\beta}_k^T\vectr{Z}}$, where $h_{0k}$ is an unspecified baseline hazard function, and $\vectr{\beta}_k$ is a vector of regression coefficients. Likewise, $h_{23}(t|\vectr{Z},t_1)=h_{023}(t)e^{\vectr{\beta}_{23}^T(\vectr{Z}^T,t_1)^T}$, where $t>t_1$. 
	For ease of presentation we include $t_1$ as a covariate, but one can consider any known transformation of $t_1$ as well as $t_1\times \vectr{Z}$ interaction terms.
	Similarly, $H_{0k}(t)=\int_{0}^{t}h_{0k}(u)du$, $k\in\{12,13,23,C\}$ are the cumulative baseline hazard functions. Denote 
	$$\mathcal{A}_k(s,t,z) = \exp\left[\left\{H_{0k}(s) - H_{0k}(t) \right\}e^{\vectr{\beta}^T_kz}\right]  , k\in\{12,13,23,C\}.$$
	Based on Eq.(\ref{Eq:mij}) it is straightforward to verify that
	\begin{eqnarray}
		\frac{m_{ji}m_{ij}}{m_{ii}m_{jj}}  &=& \exp\left[\left(\vectr{\beta}^T_{12}\vectr{Z}_i - \vectr{\beta}^T_{12}\vectr{Z}_j\right)\left(\Delta_{1j} - \Delta_{1i}\right)
		+ \left(\vectr{\beta}^T_{13}\vectr{Z}_i - \vectr{\beta}^T_{13}\vectr{Z}_j\right)\left(\Delta_{2j} - \Delta_{2i}\right)\right] \nonumber \\
		& & \frac{\mathcal{A}_{12}\left\{V_i,V_j,Z_i\right\}\mathcal{A}_{13}\left\{V_i,V_j,Z_i\right\}}{\mathcal{A}_{12}\left\{V_i,V_j,Z_j\right\}\mathcal{A}_{13}\left\{V_i,V_j,Z_j\right\}}  \nonumber \\
		& & \frac{\mathcal{A}_{23}\left\{V_j,R_i,(Z^T_i,V_j)^T\right\}^{I(R_i>V_j)}\mathcal{A}_{23}\left\{V_i,R_j,(Z^T_j,V_i)^T\right\}^{I(R_j>V_i)}}{\mathcal{A}_{23}\left\{V_i,R_i,(Z^T_i,V_i)^T\right\}^{I(R_i>V_i)}\mathcal{A}_{23}\left\{V_j,R_j,(Z^T_j,V_j)^T\right\}^{I(R_j>V_j)}} \nonumber \\
		& &
		\exp\left[\left(\vectr{\beta}^T_{C}\vectr{Z}_i - \vectr{\beta}^T_{C}\vectr{Z}_j\right)\left(\Delta_{1i} + \Delta_{2i} - \Delta_{1j} - \Delta_{2j}\right)\right] \nonumber \\
		& & 
		\frac{\mathcal{A}_{C}\left\{V_i,R_i,Z_i\right\}^{I(V_i>R_i)}\mathcal{A}_{C}\left\{V_j,R_j,Z_j\right\}^{I(V_j>R_j)}}{\mathcal{A}_{C}\left\{V_j,R_i,Z_i\right\}^{I(V_j>R_i)}\mathcal{A}_{C}\left\{V_i,R_j,Z_j\right\}^{I(V_i>R_j)}} \nonumber \\
		& &I\{R_i < V_j\}^{1-\Delta_{1j}}I\{R_j < V_i\}^{1-\Delta_{1i}} \, . \label{Eq:Mij}
	\end{eqnarray}
	Although this expression seems cumbersome, it actually admits a fairly simple estimation procedure, as described in Section 2.3. An explicit form of Eq.(\ref{Eq:Mij}) can be found in Appendix A.1.
	
	\subsection{Cox Model - Estimation of $\beta_{12}$}
	Rather than estimating all parameters simultaneously, we propose to first estimate the nuisance parameters via PL, plug those in the pairwise pseudolikelihood, and maximize with respect to the parameters of interest. In the Cox model, prevalent observations could enhance estimation of two parameters, $\vectr{\beta}_{12}$ and $H_{012}$. The latter, however, is regarded as an extra nuisance parameter, and is estimated using the Breslow estimator \citep{breslow1972contribution} with the risk-set correction for left truncation, excluding prevalent observations. Please refer to the discussion for more details about estimation of $H_{012}$. 
	
	To estimate transition $2\rightarrow3$ parameters, it is assumed that data about time from disease onset to death is accessible, which is indeed the case in most biobanks, and the UKB in particular. In this context, denote $W_i = \min(T_{2i},C_i)$ and $\Delta_{3i}=\Delta_{1i}I(W_i=T_{2i})$, so that $\Delta_{3}$ is an indicator for whether death after disease is observed.
	
	Usage of marginal models based on likelihood $L^{\mbox{III}}$ is motivated by the minimal efficiency loss in nuisance parameter estimation. As censoring occurs only after recruitment, conditioning on the entire history up to recruitment time does not effect its estimation. In transition $1\rightarrow3$, the lost information is survival time until recruitment for all observations. However, no death events are lost. Considering there are usually many deaths without disease in large biobanks (about 33,000 deaths in the UKB), incorporating survival until recruitment is unlikely to sizably affect efficiency, if at all. Lastly, transition $2\rightarrow3$ is observed in its entirety for all incident cases, whereas for the prevalent cases the information lost is only survival from disease onset until recruitment, and again, no death event is lost.
	
	Denote $\widehat{\vectr{\beta}}_k$ as the standard PL estimators of $\vectr{\beta}_{k}$, $k\in\{13,23,C\}$, and $\widetilde{\vectr{\beta}}_{12}$ as the standard PL estimator of $\vectr{\beta}_{12}$, with the risk-set correction for delayed entry, excluding the prevalent observations. Define $\widehat{H}_{0k}$ as the risk-set corrected Breslow estimators for $H_{0k}$, $k\in\{12,13,23,C\}$.
	It should be mentioned, that while $\widehat{\vectr{\beta}}_{13}$ is consistent thanks to the risk-set correction, $\widehat{H}_{013}$ can only estimate the cumulative baseline hazard function conditionally on survival up until the minimum observed recruitment time. It implies that if $H_{013}$ is estimated based on a dataset such as the UKB, where recruitment does not start at $0$, the resultant estimator will not be consistent towards the general population  cumulative baseline hazard function of transition $1\rightarrow3$. One way to correct for this bias is by using external data from publicly available life tables, as was done in \cite{gorfine2021marginalized}. However, courtesy of the difference structure $\widehat{H}_{013}(V_i) - \widehat{H}_{013}(V_j)$ appearing throughout the pairwise pseudolikelihood, this bias cancels out, and no correction is needed. This is a unique feature of our approach not shared by likelihood $L^{\mbox{II}}$.
	
	To summarize, $\widehat{\vectr{\beta}}_{12}$ is the maximizer of the following pairwise pseudo-log-likelihood
	\begin{equation} \label{Eq:PairwiseLogLike}
		l^{pair}\left(\vectr{\beta}_{12},\widehat{\vectr{\theta}},\widehat{H}_{012}\right) = -{n \choose 2}^{-1}\sum_{i<j}\ln\left\{1+\zeta_{ij}\left(\widehat{\vectr{\theta}}\right)\eta_{ij}\left(\vectr{\beta}_{12},\widehat{H}_{012}\right)\right\} \, ,
	\end{equation}
	where $\vectr{\theta}=\{\vectr{\beta}_{13},\vectr{\beta}_{23},\vectr{\beta}_{C},H_{013},H_{023},H_{0C}\}$, 
	\begin{equation} \label{Eq:eta_ij}
		\eta_{ij}(\vectr{\beta}_{12},H_{012}) =  \exp\left[\left(\vectr{\beta}^T_{12}\vectr{Z}_i - \vectr{\beta}^T_{12}\vectr{Z}_j\right)\left(\Delta_{1j} - \Delta_{1i}\right)
		+ \left\{H_{012}(V_i) - H_{012}(V_j)\right\}\left(e^{\vectr{\beta}_{12}^T\vectr{Z}_i} - e^{\vectr{\beta}_{12}^T\vectr{Z}_j}\right)\right] \, ,
	\end{equation}
	and $\zeta_{ij}\left(\widehat{\vectr{\theta}}\right)$ is the remaining elements in Eq.(\ref{Eq:Mij}) after plugging in the estimates of $\vectr{\theta}$. In Section 3 we present a sensitivity analysis assessing how the estimation of $\vectr{\beta}_{12}$ is impacted by model misspecification for the other transitions and censoring.
	
	The number of terms in Eq.(\ref{Eq:PairwiseLogLike}) is of order $O(n^2)$, rendering the estimation procedure prohibitively expensive even for moderately-sized datasets. To address this, we adopt a subsampling approach where $K_n$ pairs are selected per observation, reducing the complexity to $O(K_nn)$. The subscript $n$ indicates that the choice of the number of pairs per observation may depend on $n$. For asymptotic guarantees, discussed in Appendix A.2, it is required that $K_n \rightarrow \infty$ as $n\rightarrow\infty$. It is assumed that the data are randomly ordered and for each observation $i\in\{1,\ldots,n\}$, we include its pairwise terms with observations $\{i+1,i+2,\ldots,i+K_n\}$ (modulo $n$), and obtain
	\begin{equation} \label{Eq:PairwiseLogLikeKn}
		l_{K_n}^{pair}\left(\vectr{\beta}_{12},\widehat{\vectr{\theta}},\widehat{H}_{012}\right) = -\frac{1}{nK_n}\sum_{i=1}^n\sum_{j=i+1}^{i+K_n}\ln\left\{1+\zeta_{ij}\left(\widehat{\vectr{\theta}}\right)\eta_{ij}\left(\vectr{\beta}_{12},\widehat{H}_{012}\right)\right\} \, .
	\end{equation}

	\subsection{Cox Model - Asymptotic Results and Variance Estimation}
	This section begins with the consistency and asymptotic normality of $\widehat{\vectr{\beta}}_{12}$, followed by a discussion on variance estimation.
	Theorems 1 and 2 address the case when all pairwise terms are used in estimation, and Corollary 1 then extends these results to the subsampling framework.
	Full proofs with the required list of assumptions are provided in Appendix A.2.
	
	Denote $\vectr{\beta}^o_{k},H^o_{0k},\vectr{\theta}^o$ as the unknown true values of $\vectr{\beta}_{k},H_{0k},\vectr{\theta}$, respectively, for $k\in\{12,13,23,C\}$, and let $\Vert\cdot\Vert_2$ denote the $l^2$ norm.
	Theorem 1 establishes the consistency of the estimator. 
	\begin{theorem} 
		Under assumptions A.1--A.6, as  $n\rightarrow \infty$, 
		$$ \Vert\widehat{\vectr{\beta}}_{12} - \vectr{\beta}^o_{12}\Vert_2 = o_p(1) \, .$$ 
	\end{theorem}
	
	Before presenting Theorem 2, addressing asymptotic normality, we provide some background. Denote $\vectr{U}(\vectr{\beta}_{12},\vectr{\theta},H_{012})$ as the pairwise pseudolikelihood score function with respect to $\vectr{\beta}_{12}$,
	\begin{equation*} 
		\vectr{U}(\vectr{\beta}_{12},\vectr{\theta},H_{012}) = \frac{\partial  l^{pair}(\vectr{\beta}_{12},\vectr{\theta},H_{012})}{\partial \vectr{\beta}^T_{12}} \, .
	\end{equation*}
	We then have 
	\begin{eqnarray*}
		\vectr{0} &=& \vectr{U}\left(\vectr{\beta}^o_{12},\vectr{\theta}^o,H^o_{012}\right) +\left\{\vectr{U}\left(\widehat{\vectr{\beta}}_{12},\vectr{\theta}^o,H^o_{012}\right) - \vectr{U}\left(\vectr{\beta}^o_{12},\vectr{\theta}^o,H^o_{012}\right)\right\} 
		\nonumber \\
		&+& \left\{\vectr{U}\left(\widehat{\vectr{\beta}}_{12},\widehat{\vectr{\theta}},\widehat{H}_{012}\right) - \vectr{U}\left(\widehat{\vectr{\beta}}_{12},\vectr{\theta}^o,H^o_{012}\right)\right\} \, .
	\end{eqnarray*}  
	It will be shown that
	$$
	\sqrt{n}\left[\vectr{U}(\vectr{\beta}^o_{12},\vectr{\theta}^o,H^o_{012}) +\left\{\vectr{U}\left(\widehat{\vectr{\beta}}_{12},\widehat{\vectr{\theta}},\widehat{H}_{012}\right) - \vectr{U}\left(\widehat{\vectr{\beta}}_{12},\vectr{\theta}^o,H^o_{012}\right)\right\} \right]=\frac{1}{\sqrt{n}}\sum_{i=1}^n\vectr{\xi}_i + o_p(1) \, ,$$
	where the $\vectr{\xi}$'s are zero-mean i.i.d random vectors with $\V(\vectr{\xi})=\matrx{\mathcal{V}}$, and thus a central limit theorem follows. Additionally, as defined in assumption A.7 in Appendix A.2, $\matrx{Q}_{\beta_{12}}$ is the limiting matrix of the Hessian based on Eq.(\ref{Eq:PairwiseLogLike}), namely, as  $n\rightarrow\infty$,
	$$\frac{\partial^2 l^{pair}(\vectr{\beta}_{12},\vectr{\theta},H_{012})}{\partial \vectr{\beta}_{12}^T \partial \vectr{\beta}_{12}} \xrightarrow{p} \matrx{Q}_{\beta_{12}}(\vectr{\beta}_{12},\vectr{\theta},H_{012}) \, .$$ 
	Using a Taylor expansion for $\vectr{U}(\widehat{\vectr{\beta}}_{12},\vectr{\theta}^o,H^o_{012})$ around $\vectr{\beta}^o_{12}$, Theorem 2 will follow, and the complete proof is given in Appendix A.2.
	
	\begin{theorem} Under assumptions A.1--A.7 and as $n\rightarrow\infty$ it follows that 
		$\sqrt{n}(\widehat{\vectr{\beta}}_{12}-\vectr{\beta}^o_{12}) \xrightarrow{D} N\left(\vectr{0},\matrx{Q}^{-1}_{\beta_{12}}\matrx{\mathcal{V}}\matrx{Q}^{-1}_{\beta_{12}}\right) $, and $\matrx{Q}_{\beta_{12}}$ is evaluated at the true parameter values, namely $\matrx{Q}_{\beta_{12}}(\vectr{\beta}^o_{12},\vectr{\theta}^o,H^o_{012})$.
	\end{theorem}

	\begin{corollary}
		As $K_n\rightarrow\infty$ and $n\rightarrow\infty$, Theorems 1 and 2 extend to the subsampling framework.
	\end{corollary}
	
	Deriving a closed-form expression for $\matrx{\mathcal{V}}$ is intractable due to the nuisance parameter estimation. Thus, we present three bootstrap methods for variance estimation, preceded by introducing some additional notation.
	Denote $Y_{1i}(t) = I(R_i\le t \le V_i)$ as the at-risk process adjusted to delayed entry, and $Y_{2i}(t) = \Delta_{1i}I(\max(R_i,V_i)\le t \le W_i)$ as the at-risk process for transition  $2\rightarrow3$. Denote $\widetilde{\vectr{Z}}=\left(\vectr{Z}^T,t_1\right)^T$, and for $j=0,1,2,$ let $\vectr{S}_1^{(j)}(\vectr{\beta},t) = \sum_{i = 1}^{n}Y_{1i}(t)e^{\vectr{\beta}^T\vectr{Z}_i}\vectr{Z}_i^{\otimes j}$, and $\vectr{S}_{2}^{(j)}(\vectr{\beta},t) = \sum_{i = 1}^{n}Y_{2i}(t)e^{\vectr{\beta}^T\widetilde{\vectr{Z}}_i}\widetilde{\vectr{Z}}_i^{\otimes j}$, where $\vectr{Z}^{\otimes0}=1$, $\vectr{Z}^{\otimes1}=\vectr{Z}$ and $\vectr{Z}^{\otimes2}=\vectr{Z}\vectr{Z}^T$. Given a vector $\vectr{\omega}$ of $n$ non-negative weights, denote $\vectr{S}_{\vectr{\omega},1}^{(j)}(\vectr{\beta},t) = \sum_{i = 1}^{n}\omega_iY_{1i}(t)e^{\vectr{\beta}^T\vectr{Z}_i}\vectr{Z}_i^{\otimes j}$, and $\vectr{S}_{\omega,2}^{(j)}(\vectr{\beta},t) = \sum_{i = 1}^{n}\omega_iY_{2i}(t)e^{\vectr{\beta}^T\widetilde{\vectr{Z}}_i}\widetilde{\vectr{Z}}_i^{\otimes j}$.  \\
	\noindent \underline{\textbf{Bootstrap 1}}:
	A straightforward approach is the weighted bootstrap for U-statistics, described in Algorithm 1. Consistency of this approach follows from \cite{janssen1994weighted}, together with known consistency results of the PL-based estimators \citep{andersen1982cox}, as well as Theorems 1--2 and Corollary 1. This approach, however, entails running within each bootstrap sample the optimization routines of both the pairwise pseudolikelihood, and the PL of all transitions. In order to circumvent the latter, we propose Bootstrap 2. Bootstrap 1 is included in the simulations for comparison. \\
		\begin{algorithm}
		\spacingset{0.8}
		\caption{Full Weighted Bootstrap}
		\begin{algorithmic}
			\For{$b = 1,\ldots,B$}
			\begin{enumerate}[label=(\roman*)]
				
				\item Sample $n$ independent random weights $w^{(b)}_1,\ldots,w^{(b)}_n$ from a standard exponential distribution.
				\item Use the weights from Step (i) to solve weighted PL-based estimating equations and obtain $\widetilde{\vectr{\beta}}^{(b)}_{12},\widehat{\vectr{\beta}}^{(b)}_{13},\widehat{\vectr{\beta}}^{(b)}_{23},\widehat{\vectr{\beta}}^{(b)}_C$,
				$$\sum_{i = 1}^n\omega^{(b)}_i\Delta_{1i}\left\{\vectr{Z}_i - \frac{\vectr{S}^{(1)}_{\vectr{\omega}^{(b)},1}(\vectr{\beta}_{12},V_i)}{\vectr{S}^{(0)}_{\vectr{\omega}^{(b)},1}(\vectr{\beta}_{12},V_i)} \right\}=\vectr{0} \, \, , \sum_{i = 1}^n\omega^{(b)}_i\Delta_{2i}\left\{\vectr{Z}_i - \frac{\vectr{S}^{(1)}_{\vectr{\omega}^{(b)},1}(\vectr{\beta}_{13},V_i)}{\vectr{S}^{(0)}_{\vectr{\omega}^{(b)},1}(\vectr{\beta}_{13},V_i)} \right\}=\vectr{0}$$
				$$\sum_{i = 1}^n\omega^{(b)}_i\Delta_{3i}\left\{\vectr{Z}_i - \frac{\vectr{S}^{(1)}_{\vectr{\omega}^{(b)},2}(\vectr{\beta}_{23},W_i)}{\vectr{S}^{(0)}_{\vectr{\omega}^{(b)},2}(\vectr{\beta}_{23},W_i)} \right\}=\vectr{0}$$
				$$\sum_{i = 1}^n\omega^{(b)}_i(1-\Delta_{1i}-\Delta_{2i})\left\{\vectr{Z}_i - \frac{\vectr{S}^{(1)}_{\vectr{\omega}^{(b)},1}(\vectr{\beta}_{C},V_i)}{\vectr{S}^{(0)}_{\vectr{\omega}^{(b)},1}(\vectr{\beta}_{C},V_i)} \right\}=\vectr{0}$$
				
				. 
				\item Derive $\widehat{H}^{(b)}_{0k}$, $k\in\{12,13,23,C\}$, using weighted sums in the respective Breslow estimators, namely, 
				$$\widehat{H}^{(b)}_{012}(t) = \sum_{i = 1}^n\frac{w^{(b)}_i\Delta_{1i}I(R_i\le V_i\le t)}{\vectr{S}^{(0)}_{\vectr{\omega}^{(b)},1}\left(\widetilde{\vectr{\beta}}^{(b)}_{12},V_i\right)} \, \, ,
				\widehat{H}^{(b)}_{013}(t) = \sum_{i = 1}^n\frac{w^{(b)}_i\Delta_{2i}I(R_i\le V_i\le t)}{\vectr{S}^{(0)}_{\vectr{\omega}^{(b)},1}\left(\widehat{\vectr{\beta}}^{(b)}_{13},V_i\right)} \, ,$$
				$$\widehat{H}^{(b)}_{023}(t) = \sum_{i = 1}^n\frac{w^{(b)}_i\Delta_{3i}I(\max(V_i,R_i)\le W_i\le t)}{\vectr{S}^{(0)}_{\vectr{\omega}^{(b)},2}\left(\widehat{\vectr{\beta}}^{(b)}_{23},W_i\right)} \, ,$$
				$$\widehat{H}^{(b)}_{0C}(t) = \sum_{i = 1}^n\frac{w^{(b)}_i(1-\Delta_{1i}-\Delta_{2i})I(R_i\le V_i\le t)}{\vectr{S}^{(0)}_{\vectr{\omega}^{(b)},1}\left(\widehat{\vectr{\beta}}^{(b)}_{C},V_i\right)} \, .$$
				\item Derive
				$$\widehat{\vectr{\beta}}^{(b)}_{12} = \arg\min_{\vectr{\beta}_{12}} \frac{1}{nK_n}\sum_{i=1}^n\sum_{j=i+1}^{i+K_n}\omega_i^{(b)}\omega_j^{(b)}\ln\left\{1+\zeta_{ij}\left(\widehat{\vectr{\theta}}^{(b)}\right)\eta_{ij}\left(\vectr{\beta}_{12},\widehat{H}^{(b)}_{012}\right)\right\} $$
			\end{enumerate}
			\EndFor \\
			\Return the empirical variance matrix of $\widehat{\vectr{\beta}}^{(b)}_{12}$, $b=1,\ldots,B$.
		\end{algorithmic}
	\end{algorithm}
	\noindent \underline{\textbf{Bootstrap 2}}:
This approach relies on the factorization of likelihood $L^{\mbox{III}}$ into multiplicative components for each transition, as described in Section 2, implying that the respective maximum likelihood, or PL estimators are asymptotically independent. We propose using the asymptotic distribution of PL estimators and employ a hybrid bootstrap approach that avoids the need for nuisance parameter estimation within each bootstrap sample. This is in fact the so-called ``piggyback bootstrap", developed and theoretically justified by \cite{dixon2005functional}. Bootstrap 2 can be schematized like Algorithm 1, after replacing Step (ii) with 
	\begin{enumerate}[label=(\roman*)]
		\item[(ii)] Sample $\widetilde{\vectr{\beta}}^{(b)}_{12}$, $\widehat{\vectr{\beta}}^{(b)}_k$, $k\in\{13,23,C\}$, from	 normal distributions with means $\widetilde{\vectr{\beta}}_{12}$, $\widehat{\vectr{\beta}}_k$ and PL-based inverse information matrices as variances, see Appendix A.3 for explicit expressions.
	\end{enumerate}
		 Although faster than a full weighted-bootstrap, it still necessitates maximizing the pairwise pseudolikelihood in each bootstrap sample. Subsequently, we outline an alternative heuristic approach, only partly backed up theoretically, yet effective in practice. Importantly, this approach eliminates the need for a numerical optimization routine. \\
	\noindent \underline{\textbf{Bootstrap 3}}: This approach takes advantage of the closed-form variance formula available when the nuisance parameters are assumed known and not estimated.  The description here aligns with the subsampling framework using $K_n$ pairs per observation. The necessary modifications  for the estimator involving all pairs are outlined in Appendix A.3, which also includes the derivations leading to the final variance estimator, now being presented. Denote 
	\begin{equation} \label{Eq:scorebeta12}
		\vectr{U}_{K_n}(\vectr{\beta}_{12},\vectr{\theta},H_{012}) = \frac{\partial  l_{K_n}^{pair}(\vectr{\beta}_{12},\vectr{\theta},H_{012})}{\partial \vectr{\beta}^T_{12}} = \frac{1}{nK_n }\sum_{i=1}^n\sum_{j =i+1}^{i+K_n}\vectr{\psi}_{ij}(\vectr{\beta}_{12},\vectr{\theta},H_{012})
	\end{equation}
	where 
	\begin{equation*}
		\vectr{\psi}_{ij}(\vectr{\beta}_{12},\vectr{\theta},H_{012})=-\frac{\zeta_{ij}(\vectr{\theta})\vectr{\eta}'_{ij}(\vectr{\beta}_{12},H{_{012}})}{1+\zeta_{ij}(\vectr{\theta})\eta_{ij}(\vectr{\beta}_{12},H_{012})} \, ,
	\end{equation*}
	and
	\begin{eqnarray} 
	\vectr{\eta}'_{ij}(\vectr{\beta}_{12},H_{012}) &=& \frac{\partial  \eta_{ij}(\vectr{\beta}_{12},H_{012})}{\partial \vectr{\beta}^T_{12}} =   \eta_{ij}(\vectr{\beta}_{12},H_{012})\bigg[\left(\vectr{Z}_i - \vectr{Z}_j\right)(\Delta_{1j} - \Delta_{1i}) \nonumber  \\
	& & + \left\{H_{012}(V_i) - H_{012}(V_j)\right\}\left(e^{\vectr{\beta}_{12}^T\vectr{Z}_i}\vectr{Z}_i - e^{\vectr{\beta}_{12}^T\vectr{Z}_j}\vectr{Z}_j\right)\bigg]  \, . \label{Eq:eta_grad} 
\end{eqnarray}
	
	Then, the variance of $\widehat{\vectr{\beta}}_{12}$ can be consistently estimated by
	$$\widehat{\V}\left(\widehat{\vectr{\beta}}_{12}\right) = \widehat{\matrx{V}}_1^{-1}\widehat{\matrx{V}}_2\widehat{\matrx{V}}_1^{-1} + \widehat{\matrx{V}}_3 \, ,$$
	where
	$$\widehat{\matrx{V}}_1 = \frac{\partial\vectr{U}_{K_n}\left(\vectr{\beta}_{12},\widehat{\vectr{\theta}},\widehat{H}_{012}\right)}{\partial \vectr{\beta}_{12}}\bigg|_{\vectr{\beta}_{12} = \widehat{\vectr{\beta}}_{12}}  = \frac{\partial\vectr{U}_{K_n}\left(\widehat{\vectr{\beta}}_{12},\widehat{\vectr{\theta}},\widehat{H}_{012}\right)}{\partial \vectr{\beta}_{12}} \, ,
	$$
	and this abuse of notation recurs throughout this paper. Additionally,
	$$\widehat{\matrx{V}}_2 = \frac{1}{n^2K_n^2}\sum_{i=1}^n\sum^{i+K_n}_{j=i+1}\widehat{\vectr{\psi}}^{\otimes2}_{ij} + \frac{2(2K_n-1)}{n^2K_n^2(K_n-1)}\sum_{i = 1}^n\sum_{j = i+1}^{i+K_n}\sum^{i+K_n}_{\substack{l = i+1 \\ j\ne l}}\widehat{\vectr{\psi}}_{ij}\widehat{\vectr{\psi}}^T_{il} \, ,$$
	where $\widehat{\vectr{\psi}}_{ij}$ is in the sense of $\vectr{\psi}_{ij}\left(\widehat{\vectr{\beta}}_{12},\widehat{\vectr{\theta}},\widehat{H}_{012}\right)$. 
	
	Since deriving $\widehat{\matrx{V}}_2$ requires $O\left(nK_n^2\right)$ terms, one may wish to perform a second round of subsampling just for the sake of variance estimation. Suppose that for variance estimation one used $\widetilde{K}_n$ pairs such that $\widetilde{K}_n < K_n$, then the estimator should be modified to
	\begin{equation} \label{Eq:V2estFewerPairs}
		\widetilde{\matrx{V}}_2 = \frac{1}{n^2K_n\widetilde{K}_n}\sum_{i=1}^n\sum^{i+\widetilde{K}_n}_{j=i+1}\widehat{\vectr{\psi}}^{\otimes2}_{ij} + \frac{2(2K_n-1)}{n^2K_n\widetilde{K}_n\left(\widetilde{K}_n-1\right)}\sum_{i = 1}^n\sum_{j = i+1}^{i+\widetilde{K}_n}\sum^{i+\widetilde{K}_n}_{\substack{l = i+1 \\ j\ne l}}\widehat{\vectr{\psi}}_{ij}\widehat{\vectr{\psi}}^T_{il} \, .
	\end{equation}
	
	For $\widehat{\matrx{V}}_3$, let us generate $B$ bootstrap replicates of $\widehat{\vectr{\theta}}$ and $\widehat{H}_{012}$ following Steps (i)--(iii) in Bootstrap 2, then derive
	$$\vectr{\mathfrak{U}}^{(b)} = \left\{\frac{\partial\vectr{U}_{K_n}\left(\widehat{\vectr{\beta}}_{12},\widehat{\vectr{\theta}}^{(b)},\widehat{H}^{(b)}_{012}\right)}{\partial \vectr{\beta}_{12}}\right\}^{-1}\vectr{U}_{K_n}\left(\widehat{\vectr{\beta}}_{12},\widehat{\vectr{\theta}}^{(b)},\widehat{H}^{(b)}_{012}\right) \, ,$$
	$b=1,\ldots,B$, and $\widehat{\matrx{V}}_3$ is the empirical variance matrix estimated from these vectors. 
	 
	 The performance of the three bootstrap methods is demonstrated in the simulation study in Section 3, as well as in the real data analysis in Section 4. It is clearly seen that the methods agree with each other, and can be used for valid statistical inference. Nonetheless, within our simulation study, we encountered sporadic instability issues with Bootstrap 3 in a particular setting (setting A) under the smaller sample size scenario ($n=1,500$, and see Table \ref{Tab:simEvents} for observed-event counts), see Section 3 for more details. Therefore, as the computational burden is not heavy in small sample sizes, we would recommend Bootstrap 2 as a more suitable alternative. In contrast, when dealing with larger sample sizes, no such issue has been observed with Bootstrap 3, and it is therefore recommended, given its speed and scalability.
	
	\section{Simulation Study}
	To assess the proposed estimator's performance, a simulation study was conducted based on 200 samples, with two considered sample sizes $n=1,500/10,000$, and with $K_n = 50$. For each observation we sample its age at recruitment, censoring, disease onset, and pre-disease death.  If $\Delta_{1}=1$, we substitute the pre-disease death age with a newly-sampled post-disease death age. In this manner a large pool of observations is generated, out of which we draw $n$ observations satisfying the condition $T_2>R$. Eight covariates were generated and employed in estimating all considered models, even if not all were used for data generation. Three settings were considered, representing different data characteristics, as follows.
	
	\textbf{\underline{Setting A:}} The failure times were sampled from Cox models, with baseline hazard functions
	$h^o_{012}(t) = 0.02$, $h^o_{013}(t) = 0.02$, $h^o_{023}(t) = 0.05$, and coefficients, $\vectr{\beta}^o_{12} = (2, -1.5,$ $0.1,$ $-0.5,$ $1,$ $-2.5,$ $-1, 0)^T$, $\vectr{\beta}^o_{13} = (0.3,0,0,0,-0.2,0.4,0,0.7)^T$ and  $\vectr{\beta}^o_{23} =$ $(0, 0, 0, 0, 0,$ $0, -0.3, 0.9, 0.05)^T$, 
	where the last element in $\vectr{\beta}^o_{23}$ is the coefficient corresponding to $t_1$. 
	Denote $x_{[l]}$ as the $l$'th element of a vector $\vectr{x}$. To mimic real data where covariates may come from many dissimilar distributions, they were generated independently as follows. $Z_{[1]}$ is generated from a gamma distribution with shape 2 and rate 6, $Z_{[2]}$ from a geometric distribution with probability 1/10, $Z_{[3]}$ from an exponential distribution with rate 0.25, $Z_{[4]}$ from a beta distribution with parameters 2 and 8, $Z_{[5]}$ from a normal distribution with mean 0 and variance 4, $Z_{[6]}$ from a Weibull distribution with shape 3 and scale 4, $Z_{[7]}$ from a Poisson distribution with intensity 5 and $Z_{[8]}$ from a standard uniform distribution. 
	As a following step, each covariate was scaled to be supported on the unit interval, using the so-called ``min-max standardization", namely, given a vector $\vectr{x}$, its min-max standardization is $\vectr{x}'=\left\{\vectr{x}-\min(\vectr{x})\right\}/\left\{\max(\vectr{x})-\min(\vectr{x})\right\}$. Recruitment times were sampled from a symmetric triangular distribution between 0 and 22, and censoring times were generated from an exponential distribution with rate 0.05 restricted to be larger than the corresponding recruitment times. In this setting censoring and recruitment times are independent of the covariates. 
	
	%\begin{center}
	%	\spacingset{1}
	%	\centering
	%	\begin{tabular}{lccccccccc}
		%		& $Z_1$ & $Z_2$ & $Z_3$ & $Z_4$ & $Z_5$ & $Z_6$ & $Z_7$ & $Z_8$ & $t_1$ \\ \hline 
		%		$\vectr{\beta}^o_{12}$ & 2 & -1.5 & 0.1 & -0.5 & 1 & -2.5 & -1 & 0 & 0 \\ 
		%		$\vectr{\beta}^o_{13}$ & 0.3 & 0 & 0& 0 & -0.2 & 0.4 & 0 & 0.7 & 0 \\ 
		%		$\vectr{\beta}^o_{23}$ & 0 & 0 & 0 & 0 & 0 & 0 & -0.3 & 0.9 & 0.05 \\ 
		%		\hline 
		%	\end{tabular} 
	%\end{center}
	
	\textbf{\underline{Setting B:}} All failure and censoring times were sampled from Cox models, with coefficient vectors and baseline hazard functions identical to setting A, except for the censoring distribution which has baseline hazard function $h_C^o(t)=0.05$, and coefficient vector $\vectr{\beta}^o_C=(0,1.5,0,0,0.5,0,0,0,0)^T$.
	The censoring times were restricted to be larger than the corresponding recruitment times. Covariates were generated from a Gaussian copula with a correlation matrix having 0.8 on the off-diagonal entries, and 
	recruitment times were generated as $R = (1 + 5Z_{[1]} + 7Z_{[2]} + 10Z_{[6]} + \varepsilon)_+$, where $\varepsilon \sim N(0,1)$, and $x_+ = \max(x,0)$.
	Both the censoring and the recruitment ages depend on the covariates, but are conditionally independent of the failure times, given the covariates. Additionally, the covariates are strongly correlated.
	
	\noindent \textbf{\underline{Setting C (misspecification):}}
	Transitions $1\rightarrow3$, $2\rightarrow3$, and the censoring distribution hold secondary interest, merely serving to incorporate the prevalent observations in the analysis. Thus, assessing the estimation sensitivity to their misspecification is vital.
	Inspired by \cite{zhu2012recursively}, three models were employed to simulate transitions $1\rightarrow3$, $2\rightarrow3$ and censoring, each violating the Cox model assumptions. Despite these violations, estimation was PL-based, and the estimates were plugged into the pairwise pseudolikelihood in Eq.(\ref{Eq:PairwiseLogLikeKn}) for obtaining $\widehat{\vectr{\beta}}_{12}$.
	
	Transition $1\rightarrow2$ was simulated from a Cox model with $\vectr{\beta}^o_{12}=(2,-1,0.1,-0.5,1,-1,-1,0)^T$ and baseline hazard function $h^o_{012}(t)=0.01$. Transition $1\rightarrow3$ was generated from an exponential distribution with rate $0.04/\mu_1$, $\mu_1=\sin(\pi Z_{[1]})+2|Z_{[5]} - 0.5| + Z^3_{[6]}$, and transition $2\rightarrow3$ was generated as $T_2=G+T_1$, where $G$ is gamma-distributed with scale 3 and shape $\mu_2=0.5 + \cos(\pi Z_{[7]})^2 + 2|Z_{[8]} - 0.5| + \sqrt{T_1}/3$. Censoring ages were generated as $C=L+R$, where $R$ is the recruitment age and $L$ is generated from a log-normal distribution with $\E(\ln(L))=3|Z_{[2]}-0.5| + 2Z_{[5]}$ and $\V(\ln(L)) = 1.5^2$. Covariates were generated as in setting B, and recruitment ages were sampled such that $R=(1+5Z_{[1]}+6Z_{[2]} + 4Z_{[6]} + \varepsilon)_+$, where $\varepsilon \sim N(0,1)$.

	Table \ref{Tab:simEvents} provides observed event counts for transitions $1\rightarrow2$, $1\rightarrow3$, $2\rightarrow3$, and prevalent events. Tables \ref{Tab:sim1}--\ref{Tab:sim3} display point estimates for $\vectr{\beta}_{12}$ using the standard PL estimator with risk-set adjustment, excluding the prevalent observations, and the proposed pairwise pseudolikelihood. Empirical standard errors (SE) and relative efficiency (RE) are shown, representing the ratio of mean-squared errors between the PL and the proposed estimator. Additionally, to validate the bootstrap methods, $B=100$ bootstrap sample were used per original sample. The mean estimated SEs and coverage rates (CR) of 95\% bootstrap-based Wald-type confidence intervals are presented for all three bootstrap approaches.
	
	In all settings, point estimates closely align with the true parameters, and the proposed approach considerably outperforms PL estimators in terms of SE. In setting A, RE ranges from 1.28 to 2.04, setting B shows RE from 1.5 to 2.15, while in setting C it varies between 1.37 and 1.89. Importantly, the improvement does not diminish upon increasing the sample size. The prevalent observations account for a sizable proportion of the observed events in transition $1\rightarrow2$ (approximately 47\%, 43\% and 39\% in settings A, B, C, respectively), thus play a crucial role in the RE. All three bootstrap variance estimation approaches are in agreement, yielding close empirical and estimated SEs, while maintaining correct CRs.
	
	As noted in the previous section, in setting A with $n=1,500$, Bootstrap 3 encountered occasional instability. Among the initial set of 200 samples, 17 exhibited the presence of outlier values in at least one of their corresponding bootstrap samples. In cases where this issue arose, it typically involved only a single outlier result within the 100 bootstrap samples, though in one instance, there were as many as five such outlier results. Therefore, in setting A, with $n=1,500$, we opted to employ the established relationship between standard deviation and the median absolute deviation (MAD) for the normal distribution. In each sample within this setting, we estimated the standard errors based on Bootstrap 3 as $MAD\times 1.4826$, rather than relying on the empirical standard deviation.
	 Notably, in setting C, despite severe misspecifications, results remain robust and thus endorse the safe use of Cox models with PL-based estimation for the nuisance parameters.
	
	For sensitivity analysis on $K_n$, 200 replicates of settings A--C were generated and analyzed using $K_n=10,25,100,200$. Refer to Table \ref{Tab:simKn} in Appendix A.4 for empirical SEs. Evidently, while $K_n=10$ increased the SEs, other values merely differed, especially at $n=10,000$. These results imply that using all pairs has no extra benefit, and a modest $K_n$ value suffices.

	\FloatBarrier
	
	\begin{table}[ht]
		\spacingset{1}
		\centering
		\begin{tabular}{lccc}
			\hline
			& Setting A & Setting B & Setting C \\ 
			\hline
			$n=1,500$ & & & \\
			$n_{12}$ & 256(22)& 189(13) &  164(34) \\ 
			$n_{prev}$ & 109(12)& 81(9) &  64(8) \\
			$n_{13}$ & 484(19)& 293(15) &  352(111) \\ 
			$n_{23}$ & 186(18)& 99(11) &  102(38) \\ 
			\hline
			$n=10,000$ & & & \\
			$n_{12}$ & 1806(88) & 1252 (35) & 1092(209)  \\ 
			$n_{prev}$ & 759(44)& 542(23) &  423(21) \\ 
			$n_{13}$ & 3148(54)& 1968 (39) & 2373(720)  \\ 
			$n_{23}$ & 1310(64)& 657  (24) & 678(239)  \\ 
			\hline
		\end{tabular}
		\caption{Number of observed events per transition in the simulation study: means (standard deviations), where $n_{12}$, $n_{prev}$, $n_{13}$ and $n_{23}$ stand for the numbers of $1\rightarrow2$ (including prevalent), prevalent, $1\rightarrow3$ and $2\rightarrow3$ cases, respectively. \label{Tab:simEvents}}
	\end{table}
	
	\begin{table}[ht]
		\spacingset{1}
		\centering
		\begin{tabular}{lcccccccc}
			\hline
			$\vectr{\beta}^o_{12}$ & 2.00 & -1.50 & 0.10 & -0.50 & 1.00 & -2.50 & -1.00 & 0.00 \\ 
			\hline
			$n=1,500$&  &   &  &  &  &  &  &  \\ 
			PL & 1.96 & -1.61 & 0.10 & -0.48 & 0.96 & -2.52 & -1.03 & 0.02 \\ 
			Pairwise & 1.99 & -1.52 & 0.13 & -0.46 & 0.98 & -2.56 & -1.03 & 0.00 \\ 
			PL-SE & 0.68 & 0.95 & 0.79 & 0.55 & 0.62 & 0.55 & 0.56 & 0.27 \\ 
			Pairwise-SE & \textbf{0.52} &\textbf{ 0.72} & \textbf{0.59} & \textbf{0.42} & \textbf{0.55} & \textbf{0.44} & \textbf{0.44} & \textbf{0.21} \\ 
			RE & \textbf{1.69} & \textbf{1.78} & \textbf{1.80} & \textbf{1.69} & \textbf{1.28} &\textbf{ 1.53} & \textbf{1.64} & \textbf{1.70} \\ 
			Bootstrap1-SE & 0.54 & 0.71 & 0.63 & 0.45 & 0.52 & 0.48 & 0.49 & 0.25 \\ 
			Bootstrap2-SE & 0.56 & 0.73 & 0.64 & 0.45 & 0.51 & 0.53 & 0.50 & 0.25 \\ 
			Bootstrap3-SE & 0.54 & 0.71 & 0.62 & 0.44 & 0.49 & 0.51 & 0.49 & 0.24\\ 
			Bootstrap1-CR & 0.96 & 0.96 & 0.95 & 0.96 & 0.92 & 0.96 & 0.96 & 0.98 \\ 
			Bootstrap2-CR & 0.97 & 0.96 & 0.96 & 0.96 & 0.92 & 0.97 & 0.97 & 0.98 \\ 
			Bootstrap3-CR & 0.97 & 0.96 & 0.95 & 0.95 & 0.90 & 0.97 & 0.96 & 0.99 \\ 
			\hline
			$n=10,000$&  &   &  &  &  &  &  &  \\  
	  PL & 2.03 & -1.52 & 0.11 & -0.53 & 0.98 & -2.50 & -1.01 & 0.01 \\ 
	Pairwise & 2.02 & -1.52 & 0.10 & -0.52 & 1.01 & -2.49 & -1.01 & -0.01 \\ 
	PL-SE & 0.27 & 0.32 & 0.28 & 0.30 & 0.29 & 0.33 & 0.31 & 0.29 \\ 
	Pairwise-SE & \textbf{0.22} & \textbf{0.23} & \textbf{0.21} & \textbf{0.21} & \textbf{0.22} & \textbf{0.23} & \textbf{0.23} & \textbf{0.24} \\ 
	RE & \textbf{1.50} & \textbf{1.93} & \textbf{1.77 }& \textbf{1.92} & \textbf{1.65} & \textbf{2.04} & \textbf{1.83} & \textbf{1.51} \\ 
	Bootstrap1-SE & 0.22 & 0.24 & 0.23 & 0.23 & 0.23 & 0.25 & 0.23 & 0.23 \\ 
	Bootstrap2-SE & 0.22 & 0.24 & 0.23 & 0.23 & 0.23 & 0.25 & 0.23 & 0.23 \\ 
	Bootstrap3-SE & 0.22 & 0.24 & 0.22 & 0.23 & 0.22 & 0.25 & 0.23 & 0.23 \\ 
	Bootstrap1-CR & 0.94 & 0.94 & 0.96 & 0.98 & 0.94 & 0.98 & 0.94 & 0.93 \\ 
	Bootstrap2-CR & 0.94 & 0.95 & 0.96 & 0.98 & 0.95 & 0.98 & 0.94 & 0.93 \\ 
	Bootstrap3-CR & 0.94 & 0.95 & 0.95 & 0.97 & 0.94 & 0.97 & 0.94 & 0.92 \\ 
	\hline
		\end{tabular}
	\caption{Simulation results for setting A: point estimates based on the standard PL estimator with the risk-set adjustment for left truncation (PL) and the proposed pairwise pseudolikelihood (Pairwise), their corresponding empirical standard errors (PL-SE and Pairwise-SE) and bootstrap standard errors and coverage rates, based on 200 replicates, and $B=100$. The relative efficiency (RE) is the ratio of mean-squared errors between the PL and the proposed estimator. \label{Tab:sim1}}
	\end{table}
	
	\begin{table}[ht]
		\centering
		\spacingset{1}
		\begin{tabular}{lcccccccc}
			\hline
			$\vectr{\beta}^o_{12}$ & 2.00 & -1.50 & 0.10 & -0.50 & 1.00 & -2.50 & -1.00 & 0.00 \\ 
			\hline
			$n=1,500$ &  &  &  &  &  &  &  &  \\ 
			PL & 1.93 & -1.52 & 0.13 & -0.53 & 1.07 & -2.60 & -1.08 & 0.00 \\ 
			Pairwise & 1.91 & -1.55 & 0.18 & -0.50 & 1.02 & -2.54 & -1.08 & 0.04 \\ 
			PL-SE & 0.77 & 0.80 & 0.81 & 0.84 & 0.82 & 0.97 & 0.83 & 0.85 \\ 
			Pairwise-SE & \textbf{0.62} & \textbf{0.61} & \textbf{0.60} &\textbf{ 0.63} & \textbf{0.63} & \textbf{0.71} & \textbf{0.64} & \textbf{0.58} \\ 
			RE & \textbf{1.51} & \textbf{1.71} & \textbf{1.83} & \textbf{1.83} & \textbf{1.72} & \textbf{1.88} & \textbf{1.65} & \textbf{2.15} \\ 
			Bootstrap1-SE & 0.59 & 0.65 & 0.62 & 0.62 & 0.61 & 0.67 & 0.62 & 0.61 \\ 
			Bootstrap2-SE & 0.60 & 0.66 & 0.62 & 0.63 & 0.62 & 0.68 & 0.63 & 0.62 \\ 
			Bootstrap3-SE & 0.58 & 0.65 & 0.60 & 0.61 & 0.60 & 0.66 & 0.61 & 0.61 \\ 
			Bootstrap1-CR & 0.92 & 0.95 & 0.95 & 0.94 & 0.93 & 0.96 & 0.94 & 0.96 \\ 
			Bootstrap2-CR & 0.92 & 0.97 & 0.96 & 0.95 & 0.94 & 0.96 & 0.94 & 0.97 \\ 
			Bootstrap3-CR & 0.91 & 0.97 & 0.95 & 0.94 & 0.92 & 0.96 & 0.94 & 0.96 \\ 
			\hline
			$n=10,000$&  &   &  &  &  &  &  &  \\ 
			PL & 1.99 & -1.53 & 0.11 & -0.50 & 1.00 & -2.49 & -1.00 & 0.00 \\ 
			Pairwise & 1.99 & -1.52 & 0.11 & -0.49 & 0.99 & -2.51 & -1.00 & 0.02 \\ 
			PL-SE& 0.29 & 0.33 & 0.34 & 0.31 & 0.27 & 0.34 & 0.32 & 0.30 \\ 
			Pairwise-SE & \textbf{0.21} & \textbf{0.25} & \textbf{0.23} & \textbf{0.23} & \textbf{0.22} & \textbf{0.24} & \textbf{0.22} & \textbf{0.24} \\ 
			RE &\textbf{ 1.92} & \textbf{1.74} & \textbf{2.09} & \textbf{1.73} &\textbf{ 1.47} &\textbf{ 1.96} & \textbf{2.07} & \textbf{1.52} \\ 
			Bootstrap1-SE & 0.22 & 0.25 & 0.23 & 0.23 & 0.23 & 0.26 & 0.24 & 0.23 \\ 
			Bootstrap2-SE & 0.22 & 0.25 & 0.23 & 0.23 & 0.23 & 0.25 & 0.24 & 0.23 \\ 
			Bootstrap3-SE & 0.22 & 0.24 & 0.22 & 0.23 & 0.22 & 0.25 & 0.23 & 0.23 \\ 
			Bootstrap1-CR & 0.97 & 0.94 & 0.95 & 0.94 & 0.94 & 0.97 & 0.96 & 0.95 \\
			Bootstrap2-CR & 0.97 & 0.94 & 0.94 & 0.94 & 0.95 & 0.97 & 0.96 & 0.94 \\ 
			Bootstrap3-CR & 0.97 & 0.94 & 0.92 & 0.92 & 0.95 & 0.97 & 0.96 & 0.94 \\ 
			\hline
		\end{tabular}
		\caption{Simulation results for setting B: point estimates based on the standard PL estimator with the risk-set adjustment for left truncation (PL) and the proposed pairwise pseudolikelihood (Pairwise), their corresponding empirical standard errors (PL-SE and Pairwise-SE) and bootstrap standard errors and coverage rates, based on 200 replicates, and $B=100$. The relative efficiency (RE) is the ratio of mean-squared errors between the PL and the proposed estimator. \label{Tab:sim2}}
	\end{table}
	
	\begin{table}[ht]
		\spacingset{1}
		\centering
		\begin{tabular}{lcccccccc}
			\hline
			$\vectr{\beta}^o_{12}$ & 2.00 & -1.00 & 0.10 & -0.50 & 1.00 & -1.00 & -1.00 & 0.00 \\ 
			\hline
			$n=1,500$&  &   &  &  &  &  &  &  \\  
		  PL & 2.09 & -1.08 & 0.11 & -0.51 & 1.01 & -1.08 & -1.00 & 0.01 \\ 
		Pairwise & 2.09 & -1.09 & 0.14 & -0.53 & 1.06 & -1.07 & -0.97 & 0.05 \\ 
		PL-SE & 0.79 & 0.81 & 0.76 & 0.80 & 0.80 & 0.82 & 0.89 & 0.87 \\ 
		Pairwise-SE & \textbf{0.63} & \textbf{0.65} & \textbf{0.64} & \textbf{0.66} & \textbf{0.67} & \textbf{0.65} & \textbf{0.69} & \textbf{0.67} \\ 
		RE & \textbf{1.55} & \textbf{1.54} & \textbf{1.39} & \textbf{1.47} & \textbf{1.39} & \textbf{1.57} & \textbf{1.66} & \textbf{1.67} \\ 
		Bootstrap1-SE & 0.67 & 0.65 & 0.66 & 0.66 & 0.66 & 0.66 & 0.67 & 0.66 \\ 
		Bootstrap2-SE & 0.67 & 0.65 & 0.66 & 0.66 & 0.66 & 0.66 & 0.67 & 0.66 \\ 
		Bootstrap3-SE & 0.64 & 0.60 & 0.62 & 0.63 & 0.63 & 0.63 & 0.63 & 0.62 \\ 
		Bootstrap1-CR & 0.98 & 0.95 & 0.95 & 0.96 & 0.95 & 0.97 & 0.94 & 0.95 \\ 
		Bootstrap2-CR & 0.98 & 0.95 & 0.95 & 0.96 & 0.94 & 0.97 & 0.94 & 0.94 \\ 
		Bootstrap3-CR & 0.97 & 0.94 & 0.95 & 0.96 & 0.94 & 0.95 & 0.92 & 0.92 \\ 
		\hline
			$n=10,000$&  &   &  &  &  &  &  &  \\  
		  PL & 2.01 & -0.96 & 0.09 & -0.50 & 1.00 & -1.00 & -1.01 & -0.03 \\ 
		Pairwise & 2.02 & -0.98 & 0.10 & -0.51 & 1.02 & -0.99 & -0.96 & 0.03 \\ 
		PL-SE & 0.32 & 0.28 & 0.30 & 0.32 & 0.31 & 0.34 & 0.29 & 0.31 \\ 
		Pairwise-SE & \textbf{0.23} & \textbf{0.23} & \textbf{0.24} & \textbf{0.26} & \textbf{0.22} & \textbf{0.27} & \textbf{0.24} & \textbf{0.24} \\ 
		RE & \textbf{1.83} & \textbf{1.49} & \textbf{1.54} & \textbf{1.56} & \textbf{1.89} & \textbf{1.66} & \textbf{1.37} & \textbf{1.60} \\ 
		Bootstrap1-SE & 0.25 & 0.24 & 0.24 & 0.25 & 0.24 & 0.25 & 0.25 & 0.24 \\ 
		Bootstrap2-SE & 0.25 & 0.24 & 0.24 & 0.24 & 0.24 & 0.25 & 0.24 & 0.24 \\ 
		Bootstrap3-SE & 0.24 & 0.23 & 0.24 & 0.24 & 0.24 & 0.24 & 0.24 & 0.24 \\ 
		Bootstrap1-CR & 0.95 & 0.93 & 0.92 & 0.93 & 0.98 & 0.93 & 0.95 & 0.93 \\ 
		Bootstrap2-CR & 0.95 & 0.92 & 0.93 & 0.93 & 0.98 & 0.92 & 0.95 & 0.92 \\ 
		Bootstrap3-CR & 0.94 & 0.92 & 0.92 & 0.93 & 0.97 & 0.91 & 0.95 & 0.94 \\ 
		\hline
		\end{tabular}
		\caption{Simulation results for setting C (misspecification): point estimates based on the standard PL estimator with the risk-set adjustment for left truncation (PL) and the proposed pairwise pseudolikelihood (Pairwise), their corresponding empirical standard errors (PL-SE and Pairwise-SE) and bootstrap standard errors and coverage rates, based on 200 replicates, and $B=100$. The relative efficiency (RE) is the ratio of mean-squared errors between the PL and the proposed estimator. \label{Tab:sim3}}
	\end{table}
	
	\FloatBarrier
	
	\section{UKB - UBC Replication Study}
	We compiled a set of 31 SNPs identified in previous GWAS to be associated with UBC. Details including chromosome number, position, effect allele, other allele, and references are available in Table \ref{Tab:SNPsDetails} in Appendix A.4. The purpose is evaluating the replicability of these associations in the UKB, being an independent cohort. Individual models for each SNP were fitted, using both PL and the proposed pairwise pseudolikelihood with $K_n=100$, resulting in more than 48 million pairs. In the UKB data there are 1,761 observed events in transition $1\rightarrow2$, 637 being prevalent, and 33,059 and 602 observed events in transitions $1\rightarrow3$ and $2\rightarrow3$, respectively.
	Each model contained the SNP being examined, sex, and the first six genetic principal components to account for population substructure \citep{jeon2018determining}. SNP values and genetic PCs were standardized to have zero mean and unit variance.  For variance estimation, we employed Bootstrap 2--3 with $B=500$ bootstrap samples, and $\widetilde{K}_n=25$ in Bootstrap 3.
	
	To address multiple testing, we applied the BH procedure with a 0.05 significance threshold. All SNPs studied were previously associated with increased UBC risk, prompting one-sided tests for effects being greater than zero. Due to potential SNP correlations, their p-values might also correlate, and based on \citet[Case 1]{benjamini2001control}, it is required to confirm non-negative correlation of test statistics for validity of the BH procedure.
	To that end, 500 bootstrap samples were drawn from the UKB data, and 31 SNP-specific models were estimated using PL. The empirical correlation matrix among the resulting test statistics was then computed. The strongest negative correlation was only -0.12, whereas positive correlations neared 1, as illustrated in Figure \ref{Fig:SNPCorBoxplot} in Appendix A.4. These findings confirm non-negative correlations, validating the BH procedure. A similar conclusion for the proposed pairwise pseudolikelihood is anticipated.
	
	Analysis results are summarized in Table \ref{Tab:SNPs100a}, Table \ref{Tab:SNPs100b} in Appendix A.4, and Figure \ref{Fig:SNPse}. Figure \ref{Fig:SNPse} illustrates that the proposed approach yields lower SEs than PL, uniformly across all SNPs. Moreover, the bootstrap approaches display strong agreement regarding the estimated SEs. Owing to reduced SEs, the proposed approach revealed more significant associations, as shown in Tables \ref{Tab:SNPs100a} and \ref{Tab:SNPs100b}. Indeed, out of 31 examined SNPs, 11 achieved significance at the 0.05 level with BH correction, regardless of the chosen bootstrap procedure, in contrast to only six detected by PL. As a sensitivity analysis, we repeated the analysis with $K_n = 150$, see Table \ref{Tab:SNPs150} in Appendix A.4. Increasing $K_n$ had negligible impact on point estimates, estimated SEs, or p-values. 
	
	\FloatBarrier
	
	\begin{table}[ht]
		\spacingset{1}
		\centering
		\begin{tabular}{l|cc|cc}
			\hline
			& \multicolumn{2}{c}{PL} & \multicolumn{2}{c}{Pairwise} \\
			\hline
			SNP & est. effect & adj. p-value & est. effect & adj. p-value \\ 
			\hline
			rs11892031 & 0.052 (0.032) & 0.125 & 0.053 (0.027) & 0.059 \\ 
			rs1052133 & 0.007 (0.030) & 0.692 & 0.012 (0.025) & 0.342 \\ 
			rs10936599 & 0.028 (0.030) & 0.376 & 0.041 (0.026) & 0.115 \\ 
			rs710521 & 0.097 (0.031) & \textbf{0.009} & 0.100 (0.026) & \textbf{0.001} \\ 
			rs798766 & -0.024 (0.030) & 0.835 & 0.049 (0.025) & 0.063 \\ 
			rs401681 & 0.067 (0.030) & 0.059 & 0.077 (0.026) & \textbf{0.007} \\ 
			rs884225 & -0.012 (0.031) & 0.823 & 0.028 (0.026) & 0.263 \\ 
			rs1057868 & -0.061 (0.029) & 0.992 & -0.028 (0.025) & 0.894 \\ 
			rs17149580 & -0.017 (0.030) & 0.823 & 0.015 (0.026) & 0.342 \\ 
			rs12666814 & -0.020 (0.030) & 0.727 & 0.013 (0.025) & 0.342 \\ 
			rs73223045 & -0.014 (0.030) & 0.823 & 0.016 (0.026) & 0.342 \\ 
			rs41515546 & -0.016 (0.030) & 0.823 & 0.015 (0.026) & 0.342 \\ 
			rs12673089 & -0.015 (0.030) & 0.823 & 0.016 (0.026) & 0.342 \\ 
			rs17149628 & -0.016 (0.030) & 0.823 & 0.016 (0.026) & 0.342 \\ 
			rs17149630 & -0.016 (0.030) & 0.823 & 0.016 (0.026) & 0.342 \\ 
			rs17149636 & -0.015 (0.030) & 0.823 & 0.016 (0.026) & 0.342 \\ 
			rs1495741 & 0.063 (0.031) & 0.081 & 0.073 (0.025) & \textbf{0.007} \\ 
			rs9642880 & 0.089 (0.030) & \textbf{0.011} & 0.092 (0.026) & \textbf{0.001} \\ 
			rs2294008 & 0.056 (0.030) & 0.106 & 0.103 (0.026) & \textbf{0.001} \\ 
			rs142492877 & 0.006 (0.031) & 0.692 & 0.014 (0.026) & 0.342 \\ 
			rs907611 & 0.023 (0.030) & 0.419 & 0.024 (0.025) & 0.303 \\
			rs217727 & 0.022 (0.029) & 0.419 & -0.002 (0.025) & 0.569 \\ 
			rs9344 & -0.072 (0.030) & 0.992 & -0.041 (0.026) & 0.944 \\  
			rs4907479 & 0.084 (0.029) & \textbf{0.011} & 0.072 (0.025) & \textbf{0.007} \\ 
			rs17674580 & 0.100 (0.029) & \textbf{0.005} & 0.090 (0.025) & \textbf{0.001} \\
		rs1058396 & 0.054 (0.030) & 0.113 & 0.047 (0.025) & 0.073 \\  
			rs8102137 & 0.104 (0.029) & \textbf{0.005} & 0.081 (0.025) & \textbf{0.003} \\ 
			rs62185668 & 0.050 (0.029) & 0.123 & 0.068 (0.025) & \textbf{0.009} \\ 
			rs6104690 & 0.038 (0.030) & 0.227 & 0.025 (0.025) & 0.291 \\ 
			rs4813953 & 0.042 (0.030) & 0.193 & 0.073 (0.025) & \textbf{0.007} \\ 
			rs1014971 & 0.078 (0.031) & \textbf{0.028} & 0.067 (0.026) & \textbf{0.016} \\ 
			\hline
		\end{tabular}
		\caption{Replicability analysis of 31 SNPs based on the UKB UBC data: estimated effects (standard errors), and BH-adjusted p-values for the PL and the proposed pairwise pseudolikelihood with $K_n=100$. SEs for the pairwise pseudolikelihood are based on Bootstrap 3. Significant effects at the 0.05 threshold are marked in bold. \label{Tab:SNPs100a}}
	\end{table}

	\begin{figure}[h]
		\centering
		\spacingset{1}
		\includegraphics[width=140mm]{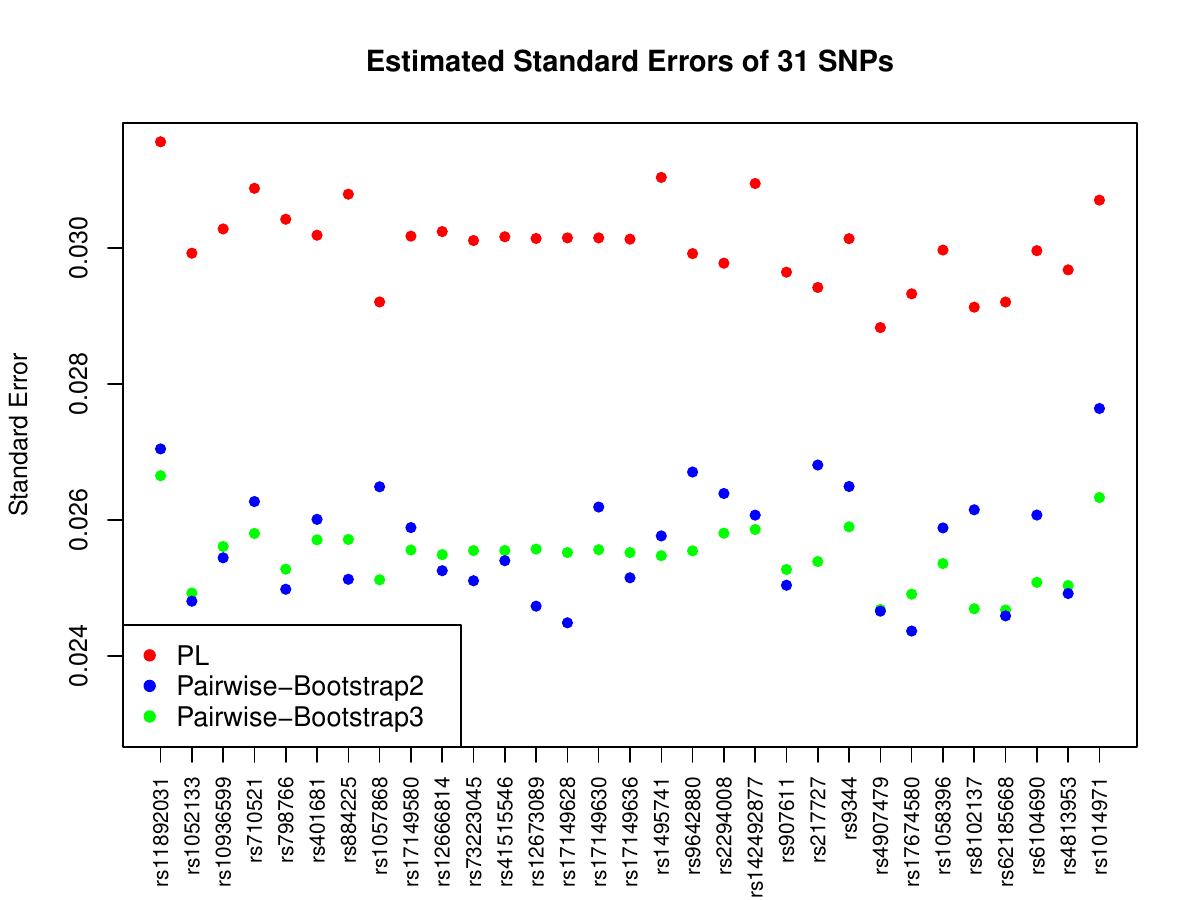}
		\caption{Replicability analysis of 31 SNPs based on the UKB UBC data: estimated standard errors based on PL (red), and Bootstrap 2--3 for the proposed approach (blue and green). \label{Fig:SNPse}}
	\end{figure}
	
	\FloatBarrier
	
	\section{Discussion}
	Existing approaches for delayed entry with prevalent observations are either statistically inefficient by disregarding prevalent information or computationally intractable due to extended runtimes and instability. Our work introduces a novel approach that substantially enhances efficiency in both statistical and computational facets.
	
	In addition to the previously-discussed issue of recall bias, which is an inherent limitation associated with prevalent data and restricts their usability in the context of time-dependent covariates, there is also one limitation in this context tied to our estimation method. The covariate trajectory of the $i$'th individual is observed until time $V_i$, so upon swapping the observed times of two observations, there will inevitably be one ``quasi-observation" with incomplete covariate trajectory. An exception is exogenous covariates, such as air-pollution levels, calendar year or weather conditions, which can be retrieved for any time point.
	
	An important application requiring only time-fixed covariates is replicability analysis for genetic variants. We used the UKB to test the replicability of previously-identified associations between 31 SNPs and UBC. The proposed approach indeed enjoyed higher statistical power compared to the vanilla PL, owing to the incorporation of the prevalent data.
	
	Theoretically, estimation of $H_{012}$ can also benefit from the prevalent observations, but the pairwise pseudolikelihood did not yield satisfactory results for this purpose. Although we have an alternative method leveraging prevalent data, its distinct tools and ideas warrant separate reporting elsewhere.
	Future work could extend the procedure to other (semi-parametric) survival models, like the accelerated failure time. Adding a penalty term to the pairwise pseudolikelihood could also be explored, necessitating adjustments to optimization and asymptotic theory.

	\begin{appendices}
		\renewcommand{\thetable}{S\arabic{table}}
		\renewcommand{\thefigure}{S\arabic{figure}}
		\renewcommand{\theequation}{S.\arabic{equation}}
		\spacingset{2}
		\setcounter{theorem}{0}
		\setcounter{table}{0}
		\setcounter{figure}{0}
		\setcounter{equation}{0}
		\setcounter{corollary}{0}
	
		\section{}
		
		\subsection{Explicit form for $m_{ji}m_{ij}/m_{ii}m_{jj}$}
		\begin{eqnarray}
			\frac{m_{ji}m_{ij}}{m_{ii}m_{jj}}  &=& \exp\left[\left(\vectr{\beta}^T_{12}\vectr{Z}_i - \vectr{\beta}^T_{12}\vectr{Z}_j\right)\left(\Delta_{1j} - \Delta_{1i}\right)
			+ \left\{H_{012}(V_i) - H_{012}(V_j)\right\}\left(e^{\vectr{\beta}_{12}^T\vectr{Z}_i} - e^{\vectr{\beta}_{12}^T\vectr{Z}_j}\right)\right] \nonumber \\
			& & \exp\left[\left(\vectr{\beta}^T_{13}\vectr{Z}_i - \vectr{\beta}^T_{13}\vectr{Z}_j\right)\left(\Delta_{2j} - \Delta_{2i}\right) +\left\{H_{013}(V_i) - H_{013}(V_j)\right\}\left(e^{\vectr{\beta}_{13}^T\vectr{Z}_i} - e^{\vectr{\beta}_{13}^T\vectr{Z}_j}\right)\right] \nonumber \\
			& & \exp\left[\left\{H_{023}(V_j) - H_{023}(R_i)\right\}e^{\vectr{\beta}^T_{23}(\vectr{Z}^T_i,V_j)^T}I\left(R_i>V_j\right)\right] \nonumber \\
			& & \exp\left[\left\{H_{023}(V_i) - H_{023}(R_j)\right\}e^{\vectr{\beta}^T_{23}(\vectr{Z}^T_j,V_i)^T}I(R_j>V_i)\right] \nonumber \\
			& & \exp\left[\left\{H_{023}(R_i) - H_{023}(V_i)\right\}e^{\vectr{\beta}^T_{23}(\vectr{Z}^T_i,V_i)^T}I(R_i>V_i)\right] \nonumber \\
			& & \exp\left[\left\{H_{023}(R_j) - H_{023}(V_j)\right\}e^{\vectr{\beta}^T_{23}(\vectr{Z}^T_j,V_j)^T}I(R_j>V_j)\right]  \,  \nonumber \\
			& &
			\exp\left[\left(\vectr{\beta}^T_{C}\vectr{Z}_i - \vectr{\beta}^T_{C}\vectr{Z}_j\right)\left(\Delta_{1i} + \Delta_{2i} - \Delta_{1j} - \Delta_{2j}\right)\right] \nonumber \\
			& & 
			\exp\left[\left\{\left(H_{0C}(V_i)-H_{0C}(R_i)\right)I(V_i>R_i) + \left(H_{0C}(R_i)- H_{0C}(V_j)\right)I(V_j>R_i)\right\}e^{\vectr{\beta}^T_{C}\vectr{Z}_i} \right] \nonumber  \\
			& & 
			\exp\left[\left\{\left(H_{0C}(V_j) - H_{0C}(R_j)\right)I(V_j>R_j) + \left(H_{0C}(R_j)- H_{0C}(V_i)\right)I(V_i>R_j)\right\}e^{\vectr{\beta}^T_{C}\vectr{Z}_j} \right] \nonumber \\
			& &I\{R_i < V_j\}^{1-\Delta_{1j}}I\{R_j < V_i\}^{1-\Delta_{1i}} \, . \nonumber
		\end{eqnarray}
		
		\subsection{Proofs}
		Before listing the required technical assumptions, denote $\tau^{(l)}_{L}$ and  $\tau^{(l)}_{U}$ as the minimum entry time and maximum follow-up time corresponding to the $l$'th at-risk process, $l=1,2$.

		\underline{Assumptions}	
		\begin{enumerate}[label=A.\arabic*]
			\item The true cumulative baseline hazard functions are bounded, namely, $\int_{0}^{\tau^{(1)}_U}\lambda^o_{0k}(t)dt<\infty$, for $k\in\{12,13,C\}$ and $\int_{0}^{\tau^{(2)}_U}\lambda^o_{023}(t)dt<\infty$. Additionally, the regression parameters $\vectr{\beta}_{k}$ lie in a compact convex set $\mathcal{B}$ of $\mathbb{R}^{p+1}$, for $k\in\{12,13,23,C\}$, that includes an open neighbourhood for each $\vectr{\beta}^o_k$. 
			\item For $l=1,2$, the functions $\vectr{s}_l^{(j)}(\vectr{\beta},t), j=0,1,2$,  defined on $\mathcal{B} \times \left[\tau^{(l)}_L,\tau_U^{(l)}\right]$,  satisfy that as $n \rightarrow \infty$, 
			$$\sup_{t\in\left[\tau^{(l)}_L,\tau_U^{(l)}\right],\vectr{\beta}\in\mathcal{B}}\frac{1}{n}\left\| \vectr{S}_l^{(j)}(\vectr{\beta},t) - \vectr{s}_l^{(j)}(\vectr{\beta},t)\right\|_2\xrightarrow{p}0 \, .$$
			
			\item For all $\vectr{\beta}\in\mathcal{B}$, $t\in\left[\tau^{(1)}_L,\tau_U^{(1)}\right]$,  
				$$\partial s_1^{(0)}(\vectr{\beta},t)/(\partial\vectr{\beta}) = \vectr{s}_1^{(1)}(\vectr{\beta},t) \, ,$$ 
				$$\partial^2s_1^{(0)}(\vectr{\beta},t)/(\partial\vectr{\beta}^T\partial\vectr{\beta}) = \matrx{s}_1^{(2)}(\vectr{\beta},t) \, ,$$
				and for all $\vectr{\beta}\in\mathcal{B}$, $t\in\left[\tau^{(2)}_L,\tau_U^{(2)}\right]$,
				$$\partial s_{2}^{(0)}(\vectr{\beta},t)/(\partial\vectr{\beta}) = \vectr{s}_{2}^{(1)}(\vectr{\beta},t) \, ,$$ 
				$$\partial^2s_{2}^{(0)}(\vectr{\beta},t)/(\partial\vectr{\beta}^T\partial\vectr{\beta}) = \matrx{s}_{2}^{(2)}(\vectr{\beta},t) \, .$$
				
				Additionally, for $j=0,1,2$, $\vectr{s}_1^{(j)}(\vectr{\beta},t)$ are continuous functions of $\vectr{\beta}$ uniformly in $t\in\left[\tau^{(1)}_L,\tau_U^{(1)}\right]$, they are bounded, and $s_1^{(0)}$ is bounded away from $0$ on $\mathcal{B} \times \left[\tau^{(1)}_L,\tau_U^{(1)}\right]$. Similarly, for $j=0,1,2$,  $\vectr{s}_{2}^{(j)}(\vectr{\beta},t)$ are continuous functions of $\vectr{\beta}$, uniformly in $t\in\left[\tau^{(2)}_L,\tau_U^{(2)}\right]$, they are bounded, and $s_{2}^{(0)}$ is bounded away from $0$ on $\mathcal{B} \times \left[\tau^{(2)}_L,\tau_U^{(2)}\right]$.
				\item  The covariates $\vectr{Z}$ are  bounded. If a transformation of $t_1$ is used as a covariate for transition $2\rightarrow3$, it should be bounded for all $t_1\in\left[0,\tau^{(1)}_U\right]$.
			\item Given the covariates, the failure times $T_1, T_2$ are conditionally independent of the censoring time $C$. Additionally, conditionally on the covariates, $T_1$ and the recruitment time $R$ are independent, and $T_2$ and $C$ are quasi-independent \citep{tsai1990testing} of $R$. 
			\item Non-emptiness of the risk sets. Namely, $\Pr\left\{Y_{li}\left(\tau^{(l)}_{L}\right)=Y_{li}\left(\tau^{(l)}_{U}\right)=1\right\}=\nu_l > 0$ for $l=1,2$ and $i=1,\ldots,n$.
			
			\item  The matrix $$\frac{\partial^2 l^{pair}(\vectr{\beta}^o_{12},\vectr{\theta}^o,H^o_{012})}{\partial \vectr{\beta}_{12}^T \partial \vectr{\beta}_{12}} $$ converges in probability to a positive definite matrix $\matrx{Q}_{\beta_{12}}(\vectr{\beta}^o_{12},\vectr{\theta}^o,H^o_{012})$. 

		\end{enumerate}
		
		Assumptions A.1--A.5 are standard regularity conditions required for the PL and Breslow estimators to be consistent for all transitions. In assumption A.1, the set $\mathcal{B}$ is assumed to lie in $\mathbb{R}^{p+1}$ when $t_1$ or a univariate transformation thereof is used as a covariate for transition $2\rightarrow3$. If a vector of covariates is created from $t_1$, or if interactions with $\vectr{Z}$ are included, the dimension of $\mathcal{B}$ should be adapted accordingly. Assumption A.6 means that there is positive probability for any observation to be at risk during the whole follow up time, namely $Y_{li}(t)=1$ for all $t\in\left[\tau^{(l)}_L,\tau^{(l)}_U\right]$, $l=1,2$. 
		
		The following proofs for Theorems 1 and 2 will first assume that all pairwise terms are involved in the estimation procedure, and no subsampling is done. Then, Corollary 1 extends these results to the subsampling case. 
		
		\underline{Consistency}
		
		\begin{theorem} \label{Eq:Theorem1}
			Under assumptions A.1--A.6, as $n\rightarrow \infty$, 
			$$ \Vert\widehat{\vectr{\beta}}_{12} - \vectr{\beta}^o_{12}\Vert_2 = o_p(1) \, .$$ 
		\end{theorem}
		
		\begin{proof}[Proof of Theorem 1]
			First, since $\vectr{\beta}_k$, $\vectr{Z}$ and $t_1\in\left[0,\tau^{(1)}_U\right]$ are bounded, see assumption A.4, there exists a constant $\kappa>0$ such that $\kappa^{-1} \le \exp\left(\vectr{\beta}_k^T\vectr{Z}\right)\le \kappa$ for all $k\in\{12,13,C\}$ and $\kappa^{-1} \le \exp\left(\vectr{\beta}_{23}^T\widetilde{\vectr{Z}}\right)\le \kappa$. Lemma 1 bounds the Breslow estimator.
			\begin{lemma}
				Under assumptions A.4 and A.6, with probability one there exists some $n^*$ such that for $n\ge n^*$, and all $t\in\left[0,\tau^{(1)}_U\right]$, $\vectr{\beta}_k\in\mathcal{B}$ 
				$$\widehat{H}_{0k}(\vectr{\beta}_k,t)\le 1.01\kappa\nu_{*}^{-1}  \, ,$$
				for $k\in\{12,13,23,C\}$, where
				$\nu_* = \min(\nu_1,\nu_2)$, and $\nu_1,\nu_2$ are defined in assumption A.6. 
			\end{lemma}
			\begin{proof}[Proof of Lemma 1]
				From the strong law of large numbers, based on assumption A.6 there exists with probability one some $n^*$ such that for all $n\ge n^*$ it holds that $$n^{-1}\sum_{i = 1}^n\min\left\{Y_{li}\left(\tau^{(l)}_L\right),Y_{li}\left(\tau^{(l)}_U\right)\right\}\ge 0.999\nu_l \, ,$$ for $l=1,2$. Let $d_k(t)$ denote the number of observed failure times of transition $k$ at time $t$, and consider the ``jump" of the Breslow estimator for $k\in\{12,13,C\}$ at some observed failure time $\widetilde{t}$
				$$\widehat{H}_{0k}\left(\widehat{\vectr{\beta}}_k,\widetilde{t}\right) - \widehat{H}_{0k}\left(\widehat{\vectr{\beta}}_k,\widetilde{t}-\right)=\frac{d_k\left(\widetilde{t}\right)}{\sum_{i = 1}^nY_{1i}\left(\widetilde{t}\right)e^{\widehat{\vectr{\beta}}_k^T\vectr{Z}_i}}\le \frac{n^{-1}\kappa d_k\left(\widetilde{t}\right)}{n^{-1}\sum_{i = 1}^n\min\left\{Y_{1i}\left(\tau^{(1)}_L\right),Y_{1i}\left(\tau^{(1)}_U\right)\right\}} \, ,$$
				so that for $n\ge n^*$ we get that the jump at time $\widetilde{t}$ is no larger than $1.01n^{-1}\kappa\nu^{-1}_1d_k\left(\widetilde{t}\right)$. Since the sum of $d_k(t)$ over all observed failure times of type $k$ cannot exceed $n$, the result follows for $k\in\{12,13,C\}$. The exact same steps can be repeated for $\widehat{H}_{023}$, using $\nu_2$, which implies the required result.
			\end{proof}
			
			Lemma 2 establishes the uniform convergence of the pseudo log-likelihood to its expectation, evaluated at the true nuisance parameter values.
			\begin{lemma} Under assumptions A.1--A.6, as $n\rightarrow\infty$, it follows that,
				\begin{equation}
					\sup_{\vectr{\beta}_{12}\in\mathcal{B}}\left| l^{pair}\left(\vectr{\beta}_{12},\widehat{\vectr{\theta}},\widehat{H}_{012}\right) - 
					\E\left\{l^{pair}(\vectr{\beta}_{12},\vectr{\theta}^o,H_{012}^o)\right\}\right| = o_p(1) \, . \label{Eq:lemma1}
				\end{equation}
			\end{lemma}
			\begin{proof}[Proof of Lemma 2]
				Let us show that the following two equations hold
				\begin{equation} \label{Eq:consis:1} 
					\sup_{\vectr{\beta}_{12}\in\mathcal{B}}\left| l^{pair}\left(\vectr{\beta}_{12},\widehat{\vectr{\theta}},\widehat{H}_{012}\right) - l^{pair}(\vectr{\beta}_{12},\vectr{\theta}^o,H^o_{012})\right| = o_p(1) \, , 
				\end{equation}
				\begin{equation} \label{Eq:consis:2}
					\sup_{\vectr{\beta}_{12}\in\mathcal{B}}\vert l^{pair}(\vectr{\beta}_{12},\vectr{\theta}^o,H_{012}^o) - \E\left\{l^{pair}(\vectr{\beta}_{12},\vectr{\theta}^o,H_{012}^o)\right\}\vert = o_p(1) \, . 
				\end{equation}
				For Eq.(\ref{Eq:consis:1}), let us first observe that although the cumulative baseline hazard functions $H_{0k}$, $k\in\{12,13,23,C\}$ are infinite-dimensional parameters, each term $L^{pair}_{ij}$ depends on them only through a finite number of terms, namely, $H_{0k}(V_i)$, $H_{0k}(V_j)$, $k\in\{12,13,23,C\}$ and $H_{023}(R_i)$, $H_{023}(R_j)$, $H_{0C}(R_i)$, $H_{0C}(R_j)$. Since $L_{ij}^{pair}$ is continuous in each of these terms, as well as in $\vectr{\beta}_k$, and since the partial likelihood and Breslow estimators are consistent, then due to the continuous mapping theorem it follows that $\left| L_{ij}^{pair}\left(\vectr{\beta}_{12},\widehat{\vectr{\theta}},\widehat{H}_{012}\right) - L_{ij}^{pair}(\vectr{\beta}_{12},\vectr{\theta}^o,H_{012}^o) \right| = o_p(1)$ for each $i\ne j$, yielding
				$\left| l^{pair}\left(\vectr{\beta}_{12},\widehat{\vectr{\theta}},\widehat{H}_{012}\right) - l^{pair}(\vectr{\beta}_{12},\vectr{\theta}^o,H^o_{012})\right| = o_p(1)$.
				The vector $\vectr{\beta}_{12}$ enters $l^{pair}$ only through the $\eta_{ij}$ terms, so by examining Eq.(\ref{Eq:eta_ij}), and due to assumptions A.1, A.4 and Lemma 1, it can be verified that the result holds over the supremum of $\vectr{\beta}_{12}$.
				
				For Eq.(\ref{Eq:consis:2}), let us note that $l^{pair}(\vectr{\beta}_{12},\vectr{\theta}^o,H_{012}^o)$ is a U-statistic, so a suitable uniform weak law of large numbers should be established. Assumptions A.1 and A.4 guarantee that $\mathcal{B}$ is compact, and that $\E\vert L_{ij}^{pair}(\vectr{\beta}_{12},\vectr{\theta}^o,H_{012}^o)\vert < \infty$ for all $\vectr{\beta}_{12}\in\mathcal{B}$, so for Eq.(\ref{Eq:consis:2}) to hold it remains to verify that $L_{ij}^{pair}(\vectr{\beta}_{12},\vectr{\theta}^o,H_{012}^o)$ is Lipschitz in $\vectr{\beta}_{12}$ \cite[corollary 4.1]{newey1991uniform}. A sufficient condition for a function to be Lipschitz is that its gradient be bounded. Based on Eq.'s (\ref{Eq:eta_ij}), (\ref{Eq:eta_grad}), assumptions A.1 and A.4, and Lemma 1, we can see that the gradient is indeed bounded, as required, and Eq.(\ref{Eq:consis:2}) holds. Finally, combining Eq.'s (\ref{Eq:consis:1})--(\ref{Eq:consis:2}) and the triangle inequality, Eq.(\ref{Eq:lemma1}) follows.
			\end{proof}
			Next, we need the following identifiability lemma.
			\begin{lemma}
				$\vectr{\beta}^o_{12}$ is the unique global maximizer of $\E\left\{l^{pair}(\vectr{\beta}_{12},\vectr{\theta}^o,H_{012}^o)\right\}$.
			\end{lemma}
			\begin{proof}[Proof of Lemma 3]
				We have that
				\begin{eqnarray*}
					& & \E\left\{l^{pair}(\vectr{\beta}_{12},\vectr{\theta}^o,H_{012}^o)\right\} = \E\left[\ln\left\{1+\zeta_{ij}(\vectr{\theta}^o)\eta_{ij}(\vectr{\beta}_{12},H^o_{012})\right\}\right] \nonumber \\ &=& \E\left[\E\left\{\ln\left(1+\zeta_{ij}(\vectr{\theta}^o)\eta_{ij}(\vectr{\beta}_{12},H^o_{012})\right)|R_i,R_j,\vectr{Z}_i,\vectr{Z}_j,R_i<T_{2i},R_j<T_{2j},(\vectr{O}_{(1)},\vectr{O}_{(2)})_{ij}\right\}\right] \, ,
				\end{eqnarray*} 
				where $(i,j)$ is a random pair, and $\vectr{\beta}^o_{12}$ is the maximizer of the inner expectation, being an expected conditional log-likelihood \citep{conniffe1987expected}, and therefore it maximizes the original expectation as well.
			\end{proof}
			The uniform convergence of $l^{pair}(\vectr{\beta}_{12},\vectr{\theta}^o,H_{012}^o)$ ensures that its continuity in $\vectr{\beta}_{12}$ carries over to its expectation. Combined with the compactness of $\mathcal{B}$ and with Lemma 3, it follows that $\vectr{\beta}^o_{12}$ is a ``well-separated" point of maximum \citep[problem 5.27]{van2000asymptotic}, and together with Lemma 2, we can invoke Theorem 5.7 of \cite{van2000asymptotic}, from which Eq.(\ref{Eq:Theorem1}) follows. 
		\end{proof}

		\underline{Normality}
		
		\begin{theorem} Under assumptions A.1--A.7, and as $n\rightarrow\infty$ it follows that 
			$\sqrt{n}\left(\widehat{\vectr{\beta}}_{12}-\vectr{\beta}^o_{12}\right) \xrightarrow{D} N(\vectr{0},\matrx{Q}^{-1}_{\beta_{12}}\matrx{\mathcal{V}}\matrx{Q}^{-1}_{\beta_{12}}) $, and $\matrx{Q}_{\beta_{12}}$ is evaluated at the true parameter values, namely $\matrx{Q}_{\beta_{12}}(\vectr{\beta}^o_{12},\vectr{\theta}^o,H^o_{012})$.
		\end{theorem}
		\begin{proof}[Proof of Theorem 2]

			We have 
			\begin{eqnarray}
				\vectr{0} &=& \vectr{U}(\vectr{\beta}^o_{12},\vectr{\theta}^o,H^o_{012}) +\left\{\vectr{U}\left(\widehat{\vectr{\beta}}_{12},\vectr{\theta}^o,H^o_{012}\right) - \vectr{U}(\vectr{\beta}^o_{12},\vectr{\theta}^o,H^o_{012})\right\} 
				\nonumber \\
				& & + \left\{\vectr{U}\left(\widehat{\vectr{\beta}}_{12},\widehat{\vectr{\theta}},\widehat{H}_{012}\right) - \vectr{U}\left(\widehat{\vectr{\beta}}_{12},\vectr{\theta}^o,H^o_{012}\right)\right\} \, . \label{Eq:scoreDecomposition}
			\end{eqnarray}  
			Based on a first-order Taylor expansion about $\vectr{\beta}^o_{12}$  we get
			\begin{eqnarray}
			 \vectr{U}\left(\widehat{\vectr{\beta}}_{12},\vectr{\theta}^o,H^o_{012}\right) - \vectr{U}(\vectr{\beta}^o_{12},\vectr{\theta}^o,H^o_{012})
				= \frac{\partial}{\partial \vectr{\beta}_{12}}\vectr{U}(\vectr{\beta}^o_{12},\vectr{\theta}^o,H^o_{012})\left(\widehat{\vectr{\beta}}_{12}-\vectr{\beta}^o_{12}\right) + \mbox{Res}\left(\breve{\vectr{\beta}}_{12}\right) \, , \label{Eq:ScoreDecompositionBeta12}
			\end{eqnarray}
			where $\breve{\vectr{\beta}}_{12}$ is on the line segment between $\widehat{\vectr{\beta}}_{12}$ and $\vectr{\beta}^o_{12}$, and the $r$'th element in the vector $\mbox{Res}\left(\breve{\vectr{\beta}}_{12}\right)$ is
			\begin{equation} \label{Eq:ScoreBeta12Res}
				\mbox{Res}_{[r]}\left(\breve{\vectr{\beta}}_{12}\right) = \left(\widehat{\vectr{\beta}}_{12}-\vectr{\beta}^o_{12}\right)^T\frac{\partial \vectr{U}'_r\left(\breve{\vectr{\beta}}_{12},\vectr{\theta}^o,H^o_{012}\right)}{\partial \vectr{\beta}^T_{12}}\left(\widehat{\vectr{\beta}}_{12}-\vectr{\beta}^o_{12}\right) \, ,
			\end{equation}
			and $\vectr{U}'_r\left(\breve{\vectr{\beta}}_{12},\vectr{\theta}^o,H^o_{012}\right)$ is the $r$'th row of the matrix $$\frac{\partial \vectr{U}\left(\breve{\vectr{\beta}}_{12},\vectr{\theta}^o,H^o_{012}\right)}{\partial \vectr{\beta}_{12}} = \frac{1}{{n \choose 2}}\sum_{i<j}-\frac{\zeta_{ij}(1+\zeta_{ij}\eta_{ij})\vectr{\eta}_{ij}'' - \zeta^2_{ij}\vectr{\eta}'^{\otimes2}_{ij}}{(1+\zeta_{ij}\eta_{ij})^2} \, ,$$
			where the arguments $\left(\breve{\vectr{\beta}}_{12},\vectr{\theta}^o,H^o_{012}\right)$ are suppressed for brevity, and 
			\begin{eqnarray*}
				\vectr{\eta}''_{ij} = \frac{\vectr{\eta}'^{\otimes2}_{ij}}{\eta_{ij}} + \eta_{ij}\left\{H^o_{012}(V_i)-H^o_{012}(V_j)\right\}\left(e^{\breve{\vectr{\beta}}^T_{12}\vectr{Z}_i}\vectr{Z}^{\otimes2}_i - e^{\breve{\vectr{\beta}}^T_{12}\vectr{Z}_j}\vectr{Z}^{\otimes2}_j  \right) \, .
			\end{eqnarray*}
			Examining a general $(l,m)$ element in the matrix $\partial \vectr{U}'_r\left(\breve{\vectr{\beta}}_{12},\vectr{\theta}^o,H^o_{012}\right)/\partial \vectr{\beta}^T_{12}$ we get
			\begin{eqnarray*}
				\frac{\partial^3 l^{pair}}{\partial \beta_{12[r]}\partial\beta_{12[l]}\partial\beta_{12[m]}} &=& -\frac{1}{{n \choose 2}}\sum_{i<j}\left\{\frac{\zeta_{ij}}{1+\zeta_{ij}\eta_{ij}} \eta'''_{ij[rlm]} \right. \\ & &- \left. \frac{\zeta_{ij}^2}{(1+\zeta_{ij}\eta_{ij})^2}\left(\eta'_{ij[r]}\eta''_{ij[lm]}+\eta'_{ij[l]}\eta''_{ij[rm]}+\eta'_{ij[m]}\eta''_{ij[rl]}\right) \right. \\ 
				& &+ \left. \frac{2\zeta^3_{ij}}{(1+\zeta_{ij}\eta_{ij})^3}\left(\eta'_{ij[r]}\eta'_{ij[l]}\eta'_{ij[m]}\right) \right\} \, ,
			\end{eqnarray*}
			where given a matrix $\matrx{X}$, $X_{[lm]}$ is its element in the $l$'th row and $m$'th column,  
			$$\eta''_{ij[rl]}=\frac{\partial^2\eta_{ij}}{\partial\beta_{12[r]}\partial\beta_{12[l]}} = \frac{\eta'_{ij[r]}\eta'_{ij[l]}}{\eta_{ij}} + \eta_{ij}(H_{012}^o(V_i)-H_{012}^o(V_j))\left(e^{\breve{\vectr{\beta}}^T_{12}\vectr{Z}_i}Z_{i[r]}Z_{i[l]} - e^{\breve{\vectr{\beta}}^T_{12}\vectr{Z}_j}Z_{j[r]}Z_{j[l]}\right) \, ,$$
			and,
			\begin{eqnarray*}
				\eta'''_{ij[rlm]}= \frac{\partial^3\eta_{ij}}{\partial\beta_{12[r]}\partial\beta_{12[l]}\partial\beta_{12[m]}}&=&\frac{\eta'_{ij[r]}\eta''_{ij[lm]}+\eta'_{ij[l]}\eta''_{ij[rm]}+\eta'_{ij[m]}\eta''_{ij[rl]}}{\eta_{ij}}-\frac{\eta'_{ij[r]}\eta'_{ij[l]}\eta'_{ij[m]}}{\eta^2_{ij}} \\
				& & +  \eta_{ij}\left\{H_{012}^o(V_i)-H_{012}^o(V_j)\right\}\left(e^{\breve{\vectr{\beta}}^T_{12}\vectr{Z}_i}Z_{i[r]}Z_{i[l]}Z_{i[m]}  \right. \\ 
				 & & \left. - e^{\breve{\vectr{\beta}}^T_{12}\vectr{Z}_j}Z_{j[r]}Z_{j[l]}Z_{j[m]}\right) \, .
			\end{eqnarray*}
			As $\breve{\vectr{\beta}}_{12}\in\mathcal{B}$, and due to assumptions A.1 and A.4, a careful inspection affirms that the matrix entries are all bounded, so based on Theorem 1 and Eq.(\ref{Eq:ScoreBeta12Res}) it follows that
			\begin{equation} \label{Eq:Beta12ResBounded}
				\mbox{Res}\left(\breve{\vectr{\beta}}_{12}\right) = O_p\left(\left\| \widehat{\vectr{\beta}}_{12} - \vectr{\beta}^o_{12}\right\|_2^2\right) = o_p\left(\left\| \widehat{\vectr{\beta}}_{12} - \vectr{\beta}^o_{12}\right\|_2\right) \, . 	
			\end{equation}
			Hence, based on Eq.'s (\ref{Eq:scoreDecomposition}), (\ref{Eq:ScoreDecompositionBeta12}), (\ref{Eq:Beta12ResBounded}) and assumption A.7, we get
			\begin{eqnarray} 
				\sqrt{n}\left(\widehat{\vectr{\beta}}_{12}-\vectr{\beta}^o_{12}\right) &=& -\matrx{Q}_{\beta_{12}}^{-1}(\vectr{\beta}^o_{12},\vectr{\theta}^o,H^o_{012})\sqrt{n}\bigg[\vectr{U}(\vectr{\beta}^o_{12},\vectr{\theta}^o,H^o_{012}) \nonumber \\
				& & +  \left\{\vectr{U}\left(\widehat{\vectr{\beta}}_{12},\widehat{\vectr{\theta}},\widehat{H}_{012}\right) - \vectr{U}\left(\widehat{\vectr{\beta}}_{12},\vectr{\theta}^o,H^o_{012}\right)\right\}\bigg]  +o_p(1) \, . \label{Eq:beta12Represent} 
			\end{eqnarray}
			Our goal now is to find an asymptotic representation of $$\sqrt{n}\left[\vectr{U}(\vectr{\beta}^o_{12},\vectr{\theta}^o,H^o_{012}) +\left\{\vectr{U}\left(\widehat{\vectr{\beta}}_{12},\widehat{\vectr{\theta}},\widehat{H}_{012}\right) - \vectr{U}\left(\widehat{\vectr{\beta}}_{12},\vectr{\theta}^o,H^o_{012}\right)\right\} \right]$$ as a sum of $n$ properly scaled i.i.d elements, and then use a central limit theorem. 
			
			\noindent First, the term $$\vectr{U}(\vectr{\beta}^o_{12},\vectr{\theta}^o,H^o_{012}) = \frac{1}{{n \choose 2}}\sum_{i<j}-\frac{\zeta_{ij}(\vectr{\theta}^o)\eta'_{ij}(\vectr{\beta}^o_{12},H{^o_{012}})}{1+\zeta_{ij}(\vectr{\theta}^o)\eta_{ij}(\vectr{\beta}^o_{12},H^o_{012})} \, ,$$ 
			is a a zero-mean U-statistic, being a score function evaluated at the true parameter values, so its H\'{a}jek projection \citep[Chapter 12]{van2000asymptotic} implies that
			\begin{equation} \label{Eq:hajekBeta12}
				\sqrt{n}\vectr{U}(\vectr{\beta}^o_{12},\vectr{\theta}^o,H^o_{012}) = \frac{2}{\sqrt{n}}\sum_{i=1}^nE\left\{\frac{\zeta_{ij}(\vectr{\theta}^o)\eta'_{ij}(\vectr{\beta}^o_{12},H{^o_{012}})}{1+\zeta_{ij}(\vectr{\theta}^o)\eta_{ij}(\vectr{\beta}^o_{12},H^o_{012})}\bigg| \vectr{O}_i,R_i,\vectr{Z}_i\right\} + o_p(1) \, .
			\end{equation} 
			As for the other term, each pairwise addend in $\vectr{U}$ depends on the cumulative baseline hazard functions only through the terms $H_{0k}(V_i)$, $H_{0k}(V_j)$, $k\in\{12,13,23,C\}$, and $H_{023}(R_i)$, $H_{023}(R_j)$, $H_{0C}(R_i)$, $H_{0C}(R_j)$. So, for each pairwise addend we tailor its Taylor expansion to the relevant terms, yielding 
			\begin{eqnarray}
				& & \sqrt{n}\left\{\vectr{U}\left(\widehat{\vectr{\beta}}_{12},\widehat{\vectr{\theta}},\widehat{H}_{012}\right) - \vectr{U}\left(\widehat{\vectr{\beta}}_{12},\vectr{\theta}^o,H^o_{012}\right)\right\} =  \label{Eq:scoreDecompositionTheta} \\
				 & & \frac{\sqrt{n}}{{n \choose 2}}\sum_{k\in\{13,23,C\}}\sum_{i<j} -\frac{\partial}{\partial \vectr{\beta}_{k}}\frac{\zeta_{ij}\left(\breve{\vectr{\theta}}^{(ij)}\right)\eta'_{ij}\left(\widehat{\vectr{\beta}}_{12},\breve{H}^{(ij)}_{012}\right)}{1+\zeta_{ij}\left(\breve{\vectr{\theta}}^{(ij)}\right)\eta_{ij}\left(\widehat{\vectr{\beta}}_{12},\breve{H}^{(ij)}_{012}\right)}\left(\widehat{\vectr{\beta}}_{k}-\vectr{\beta}^o_{k}\right) \nonumber \\ 
				&+& \frac{\sqrt{n}}{{n \choose 2}}\sum_{k\in\{12,13,23,C\}}\sum_{i<j}-\frac{\partial}{\partial (H_{0k}(V_i),H_{0k}(V_j))}\frac{\zeta_{ij}\left(\breve{\vectr{\theta}}^{(ij)}\right)\eta'_{ij}\left(\widehat{\vectr{\beta}}_{12},\breve{H}^{(ij)}_{012}\right)}{1+\zeta_{ij}\left(\breve{\vectr{\theta}}^{(ij)}\right)\eta_{ij}\left(\widehat{\vectr{\beta}}_{12},\breve{H}^{(ij)}_{012}\right)}\begin{pmatrix}
					\widehat{H}_{0k}(V_i) - H^o_{0k}(V_i) \\
					\widehat{H}_{0k}(V_j) - H^o_{0k}(V_j)
				\end{pmatrix}
				\nonumber \\ 
				&+& \frac{\sqrt{n}}{{n \choose 2}}\sum_{k\in\{23,C\}}\sum_{i<j}-\frac{\partial}{\partial (H_{0k}(R_i),H_{0k}(R_j))}\frac{\zeta_{ij}\left(\breve{\vectr{\theta}}^{(ij)}\right)\eta'_{ij}\left(\widehat{\vectr{\beta}}_{12},\breve{H}^{(ij)}_{012}\right)}{1+\zeta_{ij}\left(\breve{\vectr{\theta}}^{(ij)}\right)\eta_{ij}\left(\widehat{\vectr{\beta}}_{12},\breve{H}^{(ij)}_{012}\right)} \begin{pmatrix}
					\widehat{H}_{0k}(R_i) - H^o_{0k}(R_i) \\
					\widehat{H}_{0k}(R_j) - H^o_{0k}(R_j) 
				\end{pmatrix} \, , \nonumber
			\end{eqnarray}
			where $\breve{H}^{(ij)}_{012}$ is in the sense of $\left\{\breve{H}_{012}(V_i),\breve{H}_{012}(V_j)\right\}$ and $\breve{H}_{012}(V_l)$ is on the line segment between $\widehat{H}_{012}(V_l)$ and $H^o_{012}(V_l)$, $l=i,j$. Similarly, $\breve{\vectr{\theta}}^{(ij)}$ is in the sense of $\breve{\vectr{\beta}}_k$, $k\in\{13,23,C\}$, $\breve{H}_{0k}(V_l)$, $k\in\{13,23,C\}$, $l=i,j$, and $\breve{H}_{0k}(R_l)$, $k\in\{23,C\}$, $l=i,j$. 
			
			Denote $\matrx{Q}_{\beta_{13}}(\vectr{\beta}_{12},\vectr{\theta},H_{012})$ as the limiting matrix of $\partial \vectr{U}(\vectr{\beta}_{12},\vectr{\theta},H_{012})/\partial \vectr{\beta}_{13}$, then due to the consistency of $\widehat{\vectr{\beta}}_{12}$, which was proven in Theorem 1, and the consistency of $\widehat{\vectr{\theta}}$ and $\widehat{H}_{012}$, we have
			\begin{equation} \label{Eq:mat13conv}
				-\frac{1}{{n \choose 2}}\sum_{i<j}\frac{\partial}{\partial \vectr{\beta}_{13}}\frac{\zeta_{ij}\left(\breve{\vectr{\theta}}^{(ij)}\right)\eta'_{ij}\left(\widehat{\vectr{\beta}}_{12},\breve{H}^{(ij)}_{012}\right)}{1+\zeta_{ij}\left(\breve{\vectr{\theta}}^{(ij)}\right)\eta_{ij}\left(\widehat{\vectr{\beta}}_{12},\breve{H}^{(ij)}_{012}\right)} \xrightarrow{p} \vectr{Q}_{\beta_{13}}(\vectr{\beta}^o_{12},\vectr{\theta}^o,H^o_{012}) \, .
			\end{equation} 
			Additionally, since $\widehat{\vectr{\beta}}_{13}$ is estimated based on PL, it is a regular asymptotically linear estimator, and as such has the following asymptotic representation  \citep{tsiatis2006semiparametric}
			\begin{equation} \label{Eq:RAL13}
				\sqrt{n}\left(\widehat{\vectr{\beta}}_{13}-\vectr{\beta}^o_{13}\right) = \frac{1}{\sqrt{n}}\sum_{i=1}^n\varphi_{13}(R_i,V_i,\Delta_{2i},\vectr{Z}_i) + o_p(1) \, ,
			\end{equation}
			where $\varphi_{13}$ is known as the influence function,  defined as in \cite{reid1985influence}, but with the risk-set correction for left truncation,
			\begin{eqnarray*}
				\varphi_{13}(R_i,V_i,\Delta_{2i},\vectr{Z}_i) &=& \matrx{\Sigma}_{13}^{-1}\Delta_{2i}\left\{\vectr{Z}_i-\frac{\vectr{s}^{(1)}_{13}(\vectr{\beta}^o_{13},V_i)}{s^{(0)}_{13}(\vectr{\beta}^o_{13},V_i)}\right\} \\
				&-& \matrx{\Sigma}_{13}^{-1}e^{\vectr{Z}_i^T\vectr{\beta}^o_{13}}\int\frac{\delta_2 I(R_i\le t\le V_i)}{s^{(0)}_{13}(\vectr{\beta}^o_{13},t)}\left\{\vectr{Z}_i - \frac{\vectr{s}^{(1)}_{13}(\vectr{\beta}^o_{13},t)}{s^{(0)}_{13}(\vectr{\beta}^o_{13},t)} \right\}dF(t,\delta_2) \, ,
			\end{eqnarray*}
			where $F(t,\delta_2)$ is the joint cumulative distribution function for the observed time $V$ and the indicator $\Delta_2$, and 
			$$\matrx{\Sigma}_{13} = \int \delta_2\left[\frac{\matrx{s}^{(2)}_{13}(\vectr{\beta}^o_{13},t)}{s^{(0)}_{13}(\vectr{\beta}^o_{13},t)} - \left\{\frac{\vectr{s}^{(1)}_{13}(\vectr{\beta}^o_{13},t)}{s^{(0)}_{13}(\vectr{\beta}^o_{13},t)}\right\}^{\otimes2}\right]dF(t,\delta_2) \, .$$
			Based on Eq.'s(\ref{Eq:mat13conv})--(\ref{Eq:RAL13}), we obtain
			\begin{eqnarray} \label{Eq:QRAL13}
				 & & -\frac{\sqrt{n}}{{n \choose 2}} \sum_{i<j}\frac{\partial}{\partial \vectr{\beta}_{13}}\frac{\zeta_{ij}\left(\breve{\vectr{\theta}}^{(ij)}\right)\eta'_{ij}\left(\widehat{\vectr{\beta}}_{12},\breve{H}^{(ij)}_{012}\right)}{1+\zeta_{ij}\left(\breve{\vectr{\theta}}^{(ij)}\right)\eta_{ij}\left(\widehat{\vectr{\beta}}_{12},\breve{H}^{(ij)}_{012}\right)}\left(\widehat{\vectr{\beta}}_{13}-\vectr{\beta}^o_{13}\right)  \nonumber \\ & &  = \frac{1}{\sqrt{n}}\sum_{i=1}^n\vectr{Q}_{\beta_{13}}(\vectr{\beta}^o_{12},\vectr{\theta}^o,H^o_{012})\varphi_{13}(V_i,\Delta_{2i},\vectr{Z}_i) + o_p(1) \, .
			\end{eqnarray}
			
			The exact same steps can be taken for the terms corresponding to $\widehat{\vectr{\beta}}_{23}$ and $\widehat{\vectr{\beta}}_{C}$. 
			
			Now, denote 
			$$\matrx{W}^{(ij)}(V_i,V_j)=-\frac{\partial}{\partial (H_{012}(V_i),H_{012}(V_j))}\frac{\zeta_{ij}(\vectr{\theta}^o)\eta'_{ij}(\vectr{\beta}^o_{12},H^o_{012})}{1+\zeta_{ij}(\vectr{\theta}^o)\eta_{ij}(\vectr{\beta}^o_{12},H^o_{012})} $$
			and it will follow due to the consistency of $\widehat{\vectr{\beta}}_{12}$, $\widehat{\vectr{\theta}}$ and $\widehat{H}_{012}$, and due to the continuous mapping theorem, that 
			\begin{eqnarray}
				& &
				\frac{\sqrt{n}}{{n \choose 2}}\sum_{i<j}-\frac{\partial}{\partial (H_{012}(V_i),H_{012}(V_j))}\frac{\zeta_{ij}\left(\breve{\vectr{\theta}}\right)\eta'_{ij}\left(\widehat{\vectr{\beta}}_{12},\breve{H}_{012}\right)}{1+\zeta_{ij}\left(\breve{\vectr{\theta}}\right)\eta_{ij}\left(\widehat{\vectr{\beta}}_{12},\breve{H}_{012}\right)}\begin{pmatrix}
					\widehat{H}_{012}(V_i) - H^o_{012}(V_i) \\
					\widehat{H}_{012}(V_j) - H^o_{012}(V_j)
				\end{pmatrix} = \nonumber \\
				& &  \indent  \indent \frac{\sqrt{n}}{{n \choose 2}}\sum_{i<j}\matrx{W}^{(ij)}(V_i,V_j)\begin{pmatrix}
					\widehat{H}_{012}(V_i) - H^o_{012}(V_i) \\
					\widehat{H}_{012}(V_j) - H^o_{012}(V_j)
				\end{pmatrix} + o_p(1)\, . \label{Eq:TaylorH012}
			\end{eqnarray}
			
			Now, the notation $\matrx{W}^{(ij)}_l(t_1,t_2)$ refers to the $l$'th column of the matrix $\matrx{W}^{(ij)}$, $l=1, 2$.
			Denote $\widetilde{N}_i(t) = I(V_i\le t)$, ${N}_i(t) = \Delta_{1i}\widetilde{N}_i(t)$ and $M_{12i}(t) = N_i(t) - \int_0^tY_{1i}(u)h_{12}(u|\vectr{Z}_i)du $. Then, using the martingale representation of the Breslow estimator, we have that Eq.(\ref{Eq:TaylorH012}) is asymptotically equivalent to 
			\begin{equation*}
				\frac{\sqrt{n}}{{n \choose 2}}\sum_{i<j}	\int_0^\tau \int_0^\tau \left\{ \vectr{W}^{(ij)}_{1}(s,t)\int_{0}^s\frac{\sum_{l=1}^ndM_{12l}(u)}{\sum_{l=1}^nY_{1l}(u)e^{\vectr{\beta}^{oT}_{12}\vectr{Z}_l}} + \vectr{W}^{(ij)}_{2}(s,t)\int_{0}^t\frac{\sum_{l=1}^ndM_{12l}(u)}{\sum_{l=1}^nY_{1l}(u)e^{\vectr{\beta}^{oT}_{12}\vectr{Z}_l}}\right\}d\widetilde{N}_i(s)d\widetilde{N}_j(t) \, ,
			\end{equation*}
			which in turn, by changing the order of integration, and due to assumption A.2, is asymptotically equivalent to 
			\begin{eqnarray}
				& &
				\frac{1}{\sqrt{n}}\sum_{l=1}^n\int_{0}^\tau \frac{1}{{n \choose 2}}\sum_{i<j}\int_u^\tau\int_{0}^\tau \matrx{W}_1^{(ij)}(s,t)d\widetilde{N}_j(t)d\widetilde{N}_i(s)\frac{dM_{12l}(u)}{s^{(0)}_{12}(\vectr{\beta}^o_{12},u)} \nonumber \\
				&+& \frac{1}{\sqrt{n}}\sum_{l=1}^n\int_{0}^\tau \frac{1}{{n \choose 2}}\sum_{i<j}\int_u^\tau\int_{0}^\tau \matrx{W}_2^{(ij)}(s,t)d\widetilde{N}_i(s)d\widetilde{N}_j(t)\frac{dM_{12l}(u)}{s^{(0)}_{12}(\vectr{\beta}^o_{12},u)} \nonumber \, .
			\end{eqnarray}
			If we now denote $\vectr{\pi}_1(u)$ as the limiting value of ${n \choose 2}^{-1}\sum_{i<j}\int_u^\tau\int_{0}^\tau \matrx{W}_1^{(ij)}(s,t)d\widetilde{N}_j(t)d\widetilde{N}_i(s)$ and $\vectr{\pi}_2(u)$ as the limiting value of ${n \choose 2}^{-1}\sum_{i<j}\int_u^\tau\int_{0}^\tau \matrx{W}_2^{(ij)}(s,t)d\widetilde{N}_i(s)d\widetilde{N}_j(t)$, it will then follow that Eq.(\ref{Eq:TaylorH012}) is asymptotically equivalent to 
			\begin{equation} \label{Eq:BreslowAsympRep}
				\frac{1}{\sqrt{n}}\sum_{l=1}^n\int_{0}^\tau \frac{\vectr{\pi}_1(u)+\vectr{\pi}_2(u)}{s^{(0)}_{12}(\vectr{\beta}^o_{12},u)}dM_{12l}(u) \, ,
			\end{equation}
			which has mean zero since $M_{12l}(\cdot)$ is a zero-mean martingale for each $l=1,\ldots,n$. 
			In the same fashion, similar representations for the terms corresponding to $\widehat{H}_{013}$, $\widehat{H}_{023}$, $\widehat{H}_{0C}$ can be derived.
			
			Aggregating  Eq.'s (\ref{Eq:beta12Represent})--(\ref{Eq:scoreDecompositionTheta}), (\ref{Eq:QRAL13})--(\ref{Eq:BreslowAsympRep}) we finally obtain that $$
			\sqrt{n}\left[\vectr{U}(\vectr{\beta}^o_{12},\vectr{\theta}^o,H^o_{012}) +\left\{\vectr{U}\left(\widehat{\vectr{\beta}}_{12},\widehat{\vectr{\theta}},\widehat{H}_{012}\right) - \vectr{U}\left(\widehat{\vectr{\beta}}_{12},\vectr{\theta}^o,H^o_{012}\right)\right\} \right]=\frac{1}{\sqrt{n}}\sum_{i=1}^n\vectr{\xi}_i +o_p(1) \, ,$$
			where the $\vectr{\xi}$'s are zero-mean i.i.d random vectors, and thus a central limit theorem follows, so
			$$\frac{1}{\sqrt{n}}\sum_{i=1}^n\vectr{\xi}_i\xrightarrow{D} N(\vectr{0},\vectr{\mathcal{V}})\, ,$$ where $\matrx{\mathcal{V}}=\V(\vectr{\xi})$. Combined with Eq.(\ref{Eq:beta12Represent}) and Slutsky's theorem we finally arrive at the conclusion that 
			$$\sqrt{n}\left(\widehat{\vectr{\beta}}_{12}-\vectr{\beta}^o_{12}\right) \xrightarrow{D} N\left(\vectr{0},\matrx{Q}^{-1}_{\beta_{12}}\matrx{\mathcal{V}}\matrx{Q}^{-1}_{\beta_{12}}\right) \, ,$$
			with the true values $(\vectr{\beta}^o_{12},\vectr{\theta}^o,H^o_{012})$ inserted in $\matrx{Q}_{\beta_{12}}$.
			\end{proof}
		It should be reminded, that in practice we do not use all pairs of observations due to the high computational cost, and instead sample a number of pairs for each observation, creating a so-called incomplete U-statistic \citep{janson1984asymptotic}, as described in Section 2.3. The following corollary extends the asymptotic results to these settings.
		\begin{corollary}
			As $K_n\rightarrow\infty$ and $n\rightarrow\infty$, Theorems 1 and 2 extend to the subsampling framework.
		\end{corollary}
		\begin{proof}[Proof of Corollary 1]
			Suppose that $U_0$ is a complete U-statistic, and that $U$ is an incomplete version of it. Obviously, $\E(U)=\E(U_0)$, and due to Lemma 1 in \cite{janson1984asymptotic}, it also holds that $\E\left[\left\{\sqrt{n}(U-U_0)\right\}^2\right]=O\left(K_n^{-1}\right)$. Since $K_n\rightarrow\infty$, it will follow due to the Chebyshev inequality that $\sqrt{n}|U -U_0|\xrightarrow{p}0$, and therefore Theorems 1 and 2 will carry over for the incomplete U-statistic case. 
		\end{proof}

			\subsection{Bootstrap Methods - Additional Details}
		
		First, we give the explicit expressions for the PL-based information matrices, required for Bootstrap 2 and 3. 
	$$\matrx{\mathcal{I}}_{12} =  \sum_{i = 1}^n\Delta_{1i}\left[\frac{\matrx{S}^{(2)}_{1}\left(\widetilde{\vectr{\beta}}_{12},V_i\right)}{S^{(0)}_{1}\left(\widetilde{\vectr{\beta}}_{12},V_i\right)} -  \left\{\frac{\matrx{S}^{(1)}_{1}\left(\widetilde{\vectr{\beta}}_{12},V_i\right)}{S^{(0)}_{1}\left(\widetilde{\vectr{\beta}}_{12},V_i\right)}\right\}^{\otimes2}\right]$$
		$$\matrx{\mathcal{I}}_{13} =  \sum_{i = 1}^n\Delta_{2i}\left[\frac{\matrx{S}^{(2)}_{1}\left(\widehat{\vectr{\beta}}_{13},V_i\right)}{S^{(0)}_{1}\left(\widehat{\vectr{\beta}}_{13},V_i\right)} - \left\{\frac{\matrx{S}^{(1)}_{1}\left(\widehat{\vectr{\beta}}_{13},V_i\right)}{S^{(0)}_{1}\left(\widehat{\vectr{\beta}}_{13},V_i\right)}\right\}^{\otimes2}\right]$$
		$$\matrx{\mathcal{I}}_{23} =  \sum_{i = 1}^n\Delta_{3i}\left[\frac{\matrx{S}^{(2)}_{2}\left(\widehat{\vectr{\beta}}_{23},W_i\right)}{S^{(0)}_{2}\left(\widehat{\vectr{\beta}}_{23},W_i\right)} - \left\{\frac{\matrx{S}^{(1)}_{2}\left(\widehat{\vectr{\beta}}_{23},W_i\right)}{S^{(0)}_{2}\left(\widehat{\vectr{\beta}}_{23},W_i\right)}\right\}^{\otimes2}\right]$$
		$$\matrx{\mathcal{I}}_{C} =  \sum_{i = 1}^n(1-\Delta_{1i}-\Delta_{2i})\left[\frac{\matrx{S}^{(2)}_{1}\left(\widehat{\vectr{\beta}}_{C},V_i\right)}{S^{(0)}_{1}\left(\widehat{\vectr{\beta}}_{C},V_i\right)} - \left\{\frac{\matrx{S}^{(1)}_{1}\left(\widehat{\vectr{\beta}}_{C},V_i\right)}{S^{(0)}_{1}\left(\widehat{\vectr{\beta}}_{C},V_i\right)}\right\}^{\otimes2}\right] \, .$$
		
		For arriving at Bootstrap 3, let us use a Taylor expansion about $\vectr{\beta}^o_{12}$, and due to Theorem 1 we get
		\begin{eqnarray*}
			\vectr{0} &=& \vectr{U}_{K_n}\left(\widehat{\vectr{\beta}}_{12},\widehat{\vectr{\theta}},\widehat{H}_{012}\right)\\ 
			& =& \vectr{U}_{K_n}\left(\vectr{\beta}^o_{12},\widehat{\vectr{\theta}},\widehat{H}_{012}\right) + \frac{\partial \vectr{U}_{K_n}\left(\vectr{\beta}^o_{12},\widehat{\vectr{\theta}},\widehat{H}_{012}\right)}{\partial \vectr{\beta}_{12}}\left(\widehat{\vectr{\beta}}_{12} - \vectr{\beta}^o_{12}\right) +o_p\left(\left\|	\widehat{\vectr{\beta}}_{12} - \vectr{\beta}^o_{12} \right\|\right)
		\end{eqnarray*}
		and so
		\begin{equation*}
			\widehat{\vectr{\beta}}_{12} - \vectr{\beta}^o_{12} = -\left\{\frac{\partial\vectr{U}_{K_n}\left(\vectr{\beta}^o_{12},\widehat{\vectr{\theta}},\widehat{H}_{012}\right)}{\partial \vectr{\beta}_{12}}\right\}^{-1}\vectr{U}_{K_n}\left(\vectr{\beta}^o_{12},\widehat{\vectr{\theta}},\widehat{H}_{012}\right) +o_p\left(\left\| 	\widehat{\vectr{\beta}}_{12} - \vectr{\beta}^o_{12} \right\|\right)\, .
		\end{equation*}
		
		From the law of total variance it follows that 
		\begin{eqnarray}
			\V\left(\widehat{\vectr{\beta}}_{12} - \vectr{\beta}^o_{12}\right) &=& \E\left[\V\left\{\left(\frac{\partial\vectr{U}_{K_n}\left(\vectr{\beta}^o_{12},\widehat{\vectr{\theta}},\widehat{H}_{012}\right)}{\partial \vectr{\beta}_{12}}\right)^{-1}\vectr{U}_{K_n}\left(\vectr{\beta}^o_{12},\widehat{\vectr{\theta}},\widehat{H}_{012}\right)\bigg|\widehat{\vectr{\theta}},\widehat{H}_{012}\right\}\right] \label{Eq:BootTotalvar} \\
			&+&  \V\left[\E\left\{\left(\frac{\partial\vectr{U}_{K_n}\left(\vectr{\beta}^o_{12},\widehat{\vectr{\theta}},\widehat{H}_{012}\right)}{\partial \vectr{\beta}_{12}}\right)^{-1}\vectr{U}_{K_n}\left(\vectr{\beta}^o_{12},\widehat{\vectr{\theta}},\widehat{H}_{012}\right)\bigg|\widehat{\vectr{\theta}},\widehat{H}_{012}\right\}\right]  + o_p(1)\, \nonumber. 
		\end{eqnarray}
		Under a working assumption that $\left(\widehat{\vectr{\theta}},\widehat{H}_{012}\right)$ and $\vectr{U}_{K_n}(\vectr{\beta}^o_{12},\vectr{\theta}^o,H^o_{012})$ are independent, the inner variance in the first term can be estimated as if $\widehat{\vectr{\theta}}$ and $\widehat{H}_{012}$ were fixed, using a sandwich-type variance estimator. Namely, under this independence working assumption it can be shown that
		\begin{eqnarray*}
		\V\left\{\left(\frac{\partial\vectr{U}_{K_n}\left(\vectr{\beta}^o_{12},\widehat{\vectr{\theta}},\widehat{H}_{012}\right)}{\partial \vectr{\beta}_{12}}\right)^{-1}\vectr{U}_{K_n}\left(\vectr{\beta}^o_{12},\widehat{\vectr{\theta}},\widehat{H}_{012}\right)\bigg|\widehat{\vectr{\theta}},\widehat{H}_{012}\right\} &=& \matrx{V}^{-1}_1\left(\vectr{\beta}^o_{12},\widehat{\vectr{\theta}},\widehat{H}_{012}\right)\matrx{V}_2\left(\vectr{\beta}^o_{12},\widehat{\vectr{\theta}},\widehat{H}_{012}\right) \\ & &\matrx{V}^{-1}_1\left(\vectr{\beta}^o_{12},\widehat{\vectr{\theta}},\widehat{H}_{012}\right) \, ,
		\end{eqnarray*}
		where
		$$\matrx{V}_1(\vectr{\beta}_{12},\vectr{\theta},H_{012}) = \E\left\{\frac{\partial\vectr{U}_{K_n}(\vectr{\beta}_{12},\vectr{\theta},H_{012})}{\partial \vectr{\beta}_{12}}\right\} \, ,
		$$
		and the expectation here treats the arguments $\vectr{\beta}_{12},\vectr{\theta},H_{012}$ as fixed, so that for instance
		$$\matrx{V}_1\left(\vectr{\beta}^o_{12},\widehat{\vectr{\theta}},\widehat{H}_{012}\right) = \E\left\{\frac{\partial\vectr{U}_{K_n}(\vectr{\beta}_{12},\vectr{\theta},H_{012})}{\partial \vectr{\beta}_{12}}\right\}_{\vectr{\beta}_{12}=\vectr{\beta}^o_{12},\vectr{\theta}=\widehat{\vectr{\theta}},H_{012}=\widehat{H}_{012}} \, . $$
		Additionally,
		$$\matrx{V}_2(\vectr{\beta}_{12},\vectr{\theta},H_{012}) = \V\left\{\frac{1}{nK_n}\sum_{i=1}^n\sum^{i+K_n}_{j=i+1}\vectr{\psi}_{ij}(\vectr{\beta}_{12},\vectr{\theta},H_{012})\right\} = \frac{\V(\vectr{\psi}_{ij})}{nK_n}+\frac{2(2K_n-1)\cov(\vectr{\psi}_{ij},\vectr{\psi}_{il})}{nK_n} \, ,$$
		where $(i,j)$ and $(i,l)$ are two random pairs sharing one index in common, and similarly to $\matrx{V}_1$, the variance and covariance treat the arguments of $\vectr{\psi}$ as fixed. These matrices can be estimated by  
		$$\widehat{\matrx{V}}_1\left(\vectr{\beta}^o_{12},\widehat{\vectr{\theta}},\widehat{H}_{012}\right) = \frac{\partial\vectr{U}_{K_n}\left(\widehat{\vectr{\beta}}_{12},\widehat{\vectr{\theta}},\widehat{H}_{012}\right)}{\partial \vectr{\beta}_{12}} \, ,
		$$
		$$\widehat{\matrx{V}}_2\left(\vectr{\beta}^o_{12},\widehat{\vectr{\theta}},\widehat{H}_{012}\right) = \frac{1}{n^2K_n^2}\sum_{i=1}^n\sum^{i+K_n}_{j=i+1}\widehat{\vectr{\psi}}^{\otimes2}_{ij} + \frac{2(2K_n-1)}{n^2K_n^2(K_n-1)}\sum_{i = 1}^n\sum_{j = i+1}^{i+K_n}\sum^{i+K_n}_{\substack{l = i+1 \\ j\ne l}}\widehat{\vectr{\psi}}_{ij}\widehat{\vectr{\psi}}^T_{il} \, ,$$
		where $\widehat{\vectr{\psi}}_{ij}$ is in the sense of $\vectr{\psi}_{ij}\left(\widehat{\vectr{\beta}}_{12},\widehat{\vectr{\theta}},\widehat{H}_{012}\right)$.
		
		For estimating the second addend in Eq.(\ref{Eq:BootTotalvar}), one should observe that the inner conditional expectation is a random variable with respect to $\widehat{\vectr{\theta}}$ and $\widehat{H}_{012}$. To estimate this variance term, we can generate $B$ bootstrap replicates of $\widehat{\vectr{\theta}}$ and $\widehat{H}_{012}$ following Steps (i)--(iii) in Bootstrap 2, then derive
	$$\vectr{\mathfrak{U}}^{(b)}= \left(\frac{\partial\vectr{U}_{K_n}\left(\widehat{\vectr{\beta}}_{12},\widehat{\vectr{\theta}}^{(b)},\widehat{H}^{(b)}_{012}\right)}{\partial \vectr{\beta}_{12}}\right)^{-1}\vectr{U}_{K_n}\left(\widehat{\vectr{\beta}}_{12},\widehat{\vectr{\theta}}^{(b)},\widehat{H}^{(b)}_{012}\right) \, ,$$
	$b=1,\ldots,B$, and calculate the empirical variance matrix of these vectors. Combining the estimates for the two variance sources would thus yield an estimate for the variance of $\widehat{\vectr{\beta}}_{12}$.
	
	If the estimator uses all pairwise terms the following modifications should be made,
	\begin{equation*}
	\vectr{U}(\vectr{\beta}_{12},\vectr{\theta},H_{012}) = \frac{1}{{n \choose 2}}\sum_{i<j}\vectr{\psi}_{ij}(\vectr{\beta}_{12},\vectr{\theta},H_{012}) \, ,
	\end{equation*} 
	$$\widehat{\matrx{V}}_1\left(\vectr{\beta}^o_{12},\widehat{\vectr{\theta}},\widehat{H}_{012}\right) = \frac{\partial\vectr{U}\left(\widehat{\vectr{\beta}}_{12},\widehat{\vectr{\theta}},\widehat{H}_{012}\right)}{\partial \vectr{\beta}_{12}} \, ,$$
	$$\matrx{V}_2(\vectr{\beta}_{12},\vectr{\theta},H_{012}) = \V\left\{\frac{1}{{n \choose 2}}\sum_{i<j}\vectr{\psi}_{ij}(\vectr{\beta}_{12},\vectr{\theta},H_{012})\right\} = \frac{1}{{n \choose 2}}\V(\vectr{\psi}_{ij})+\frac{4(n-2)}{n(n-1)}\cov(\vectr{\psi}_{ij},\vectr{\psi}_{il}) \, ,$$
	$$\widehat{\matrx{V}}_2(\vectr{\beta}_{12},\vectr{\theta},H_{012}) = \frac{1}{{n \choose 2}^2}\sum_{i<j}\widehat{\vectr{\psi}}^{\otimes2}_{ij} + \frac{4}{n^2(n-1)^2}\sum_{i = 1}^n\sum_{\substack{j\ne l \\ j,l\ne i}}\widehat{\vectr{\psi}}_{ij}\widehat{\vectr{\psi}}^T_{il} \, , $$
	and
	$$\vectr{\mathfrak{U}}^{(b)}= \left(\frac{\partial\vectr{U}\left(\widehat{\vectr{\beta}}_{12},\widehat{\vectr{\theta}}^{(b)},\widehat{H}^{(b)}_{012}\right)}{\partial \vectr{\beta}_{12}}\right)^{-1}\vectr{U}\left(\widehat{\vectr{\beta}}_{12},\widehat{\vectr{\theta}}^{(b)},\widehat{H}^{(b)}_{012}\right) \, ,$$
		
		\subsection{Additional Figures and Tables}
		 \FloatBarrier
		 
		\begin{table}[ht]
			\centering
			\spacingset{1}
			\begin{tabular}{cccccccccc}
				\hline
				Setting&$K_n$ & $\beta_{12_{[1]}}$ & $\beta_{12_{[2]}}$ & $\beta_{12_{[3]}}$ & $\beta_{12_{[4]}}$ & $\beta_{12_{[5]}}$ &$\beta_{12_{[6]}}$ & $\beta_{12_{([7]}}$ & $\beta_{12_{[8]}}$ \\ 
				\hline
				$n=1,500$ &&&&&&&&&\\
				A &10 & 0.538 & 0.749 & 0.589 & 0.426 & 0.559 & 0.477 & 0.465 & 0.218 \\ 
				A&25 & 0.523 & 0.712 & 0.575 & 0.417 & 0.549 & 0.453 & 0.445 & 0.211 \\ 
				A&50 & 0.525 & 0.716 & 0.588 & 0.421 & 0.548 & 0.441 & 0.440 & 0.211 \\ 
				A&100 & 0.523 & 0.713 & 0.581 & 0.415 & 0.546 & 0.432 & 0.442 & 0.211 \\ 
				A&200 & 0.524 & 0.710 & 0.577 & 0.415 & 0.543 & 0.434 & 0.443 & 0.212 \\ 
				B&10 & 0.655 & 0.642 & 0.625 & 0.655 & 0.657 & 0.756 & 0.663 & 0.576 \\ 
				B&25 & 0.628 & 0.612 & 0.598 & 0.634 & 0.632 & 0.739 & 0.642 & 0.574 \\ 
				B&50 & 0.624 & 0.614 & 0.597 & 0.625 & 0.629 & 0.710 & 0.642 & 0.576 \\ 
				B&100 & 0.618 & 0.613 & 0.590 & 0.619 & 0.627 & 0.711 & 0.633 & 0.574 \\ 
				B&200 & 0.614 & 0.615 & 0.591 & 0.613 & 0.619 & 0.703 & 0.631 & 0.572 \\ 
				C&10 & 0.674 & 0.695 & 0.674 & 0.689 & 0.692 & 0.678 & 0.710 & 0.727 \\ 
				C&25 & 0.639 & 0.656 & 0.639 & 0.665 & 0.677 & 0.664 & 0.686 & 0.682 \\ 
				C&50 & 0.633 & 0.651 & 0.644 & 0.656 & 0.674 & 0.653 & 0.691 & 0.674 \\ 
				C&100 & 0.641 & 0.645 & 0.638 & 0.641 & 0.670 & 0.646 & 0.685 & 0.674 \\ 
				C&200 & 0.634 & 0.639 & 0.632 & 0.641 & 0.663 & 0.641 & 0.677 & 0.671 \\
				\hline 
				$n=10,000$ &&&&&&&&& \\
				A&10 & 0.285 & 0.318 & 0.288 & 0.177 & 0.196 & 0.197 & 0.197 & 0.088 \\ 
				A&25 & 0.281 & 0.311 & 0.289 & 0.177 & 0.194 & 0.195 & 0.190 & 0.087 \\ 
				A&50 & 0.279 & 0.307 & 0.285 & 0.173 & 0.191 & 0.193 & 0.188 & 0.088 \\ 
				A&100 & 0.279 & 0.308 & 0.282 & 0.173 & 0.189 & 0.189 & 0.186 & 0.088 \\ 
				A&200 & 0.278 & 0.309 & 0.280 & 0.173 & 0.189 & 0.190 & 0.185 & 0.088 \\ 
				B&10 & 0.233 & 0.242 & 0.217 & 0.222 & 0.233 & 0.234 & 0.247 & 0.253 \\ 
				B&25 & 0.221 & 0.230 & 0.213 & 0.215 & 0.229 & 0.230 & 0.235 & 0.243 \\ 
				B&50 & 0.220 & 0.227 & 0.212 & 0.213 & 0.223 & 0.227 & 0.229 & 0.239 \\ 
				B&100 & 0.219 & 0.227 & 0.209 & 0.213 & 0.225 & 0.226 & 0.229 & 0.237 \\ 
				B&200 & 0.217 & 0.225 & 0.209 & 0.215 & 0.223 & 0.225 & 0.228 & 0.237 \\ 
				C&10 & 0.242 & 0.237 & 0.260 & 0.268 & 0.238 & 0.282 & 0.249 & 0.254 \\ 
				C&25 & 0.234 & 0.238 & 0.241 & 0.261 & 0.225 & 0.269 & 0.243 & 0.249 \\ 
				C&50 & 0.232 & 0.232 & 0.239 & 0.257 & 0.222 & 0.267 & 0.240 & 0.244 \\ 
				C&100 & 0.228 & 0.227 & 0.240 & 0.257 & 0.219 & 0.265 & 0.242 & 0.240 \\ 
				C&200 & 0.228 & 0.226 & 0.240 & 0.255 & 0.219 & 0.263 & 0.241 & 0.239 \\ 
				\hline
			\end{tabular}
		\caption{Simulation results: estimated standard errors of $\widehat{\vectr{\beta}}_{12}$ based on 200 replicates for settings A--C, and different values of $K_n$. \label{Tab:simKn}}
		\end{table}

			\begin{sidewaystable}[ht]
			\spacingset{1}
			\centering
			\fbox{%
				\scalebox{0.85}{
					\begin{tabular}{llclccl}
						
						& RSID & Chromosome & Position & OA & EA & References \\ 
						\hline
						1 & rs11892031 &   2 & 234565283 & C & A & \cite{selinski2012rs11892031,zhang2014genetic} \\ 
						2 & rs1052133 &   3 & 9798773 & C & G & \cite{kim2005genotypes,karahalil2006dna,ma2012hogg1} \\ 
						3 & rs10936599 &   3 & 169492101 & T & C & \cite{figueroa2014genome,polat2019association} \\ 
						4 & rs710521 &   3 & 189645933 & C & T & \cite{kiemeney2008sequence,stern2009sequence,lehmann2010rs710521}  \\ 
						5 & rs798766 &   4 & 1734239 & C & T &  \cite{kiemeney2010sequence,figueroa2015modification,meng2017association}\\ 
						6 & rs401681 &   5 & 1322087 & T & C &  \cite{rafnar2009sequence,gago2011genetic}\\ 
						7 & rs884225 &   7 & 55274084 & T & C & \cite{chu2013egfr,luo2021rs884225} \\ 
						8 & rs1057868 &   7 & 75615006 & T & C & \cite{xiao2015functional} \\ 
						9 & rs17149580 &   7 & 125978216 & A & G & \cite{lipunova2019genome} \\ 
						10 & rs12666814 &   7 & 125979540 & C & T & \cite{lipunova2019genome} \\ 
						11 & rs73223045 &   7 & 125992106 & G & C & \cite{lipunova2019genome} \\ 
						12 & rs41515546 &   7 & 125998959 & T & C & \cite{lipunova2019genome} \\ 
						13 & rs12673089 &   7 & 126006133 & C & T & \cite{lipunova2019genome} \\ 
						14 & rs17149628 &   7 & 126006965 & C & T & \cite{lipunova2019genome} \\ 
						15 & rs17149630 &   7 & 126006996 & C & T & \cite{lipunova2019genome} \\ 
						16 & rs17149636 &   7 & 126018952 & A & G & \cite{lipunova2019genome} \\ 
						17 & rs1495741 &   8 & 18272881 & G & A & \cite{rothman2010multi,garcia2011single,figueroa2014genome} \\ 
						18 & rs9642880 &   8 & 128718068 & G & T & \cite{kiemeney2008sequence,wang2009common,mamdouh2022molecular} \\ 
						19 & rs2294008 &   8 & 143761931 & C & T & \cite{wu2009genetic,wang2010genetic,fu2012common,ma2013systematic} \\ 
						20 & rs142492877 &   9 & 98482828 & A & G & \cite{lipunova2019genome} \\ 
						21 & rs907611 &  11 & 1874072 & G & A &  \cite{figueroa2014genome}\\ 
						22 & rs217727 &  11 & 2016908 & G & A & \cite{hua2016genetic} \\ 
						23 & rs9344 &  11 & 69462910 & G & A &  \cite{yuan2010cyclin}\\ 
						24 & rs4907479 &  13 & 113659108 & G & A & \cite{figueroa2016identification} \\ 
						25 & rs17674580 &  18 & 43309911 & C & T & \cite{rafnar2011european,wang2014cumulative} \\ 
						26 & rs1058396 &  18 & 43319519 & A & G & \cite{rafnar2011european} \\ 
						27 & rs8102137 &  19 & 30296853 & T & C &  \cite{rothman2010multi}\\ 
						28 & rs62185668 &  20 & 10961935 & C & A & \cite{figueroa2016identification} \\ 
						29 & rs6104690 &  20 & 10988099 & G & A &  \cite{figueroa2016identification}\\ 
						30 & rs4813953 &  20 & 10991138 & C & T & \cite{rafnar2014genome}\\ 
						31 & rs1014971 &  22 & 39332623 & C & T & \cite{rothman2010multi} \\
			\end{tabular}}}
			\caption{Replicability analysis of 31 SNPs based on the UKB UBC data: additional SNP details and corresponding references. EA and OA stand for effect allele and other allele, respectively.  \label{Tab:SNPsDetails}}
		\end{sidewaystable}
		
		\begin{table}[ht]
			\spacingset{1}
			\centering
			\begin{tabular}{l|cc}
				\hline
				& \multicolumn{2}{c}{Pairwise}     \\ 
				\hline
				SNP &  est. effect  & adj. p-value   \\ 
				\hline
				rs11892031 & 0.053 (0.027) & 0.060 \\ 
				rs1052133 & 0.012 (0.025) & 0.341 \\ 
				rs10936599 & 0.041 (0.025) & 0.113 \\ 
				rs710521 & 0.100 (0.026) & \textbf{0.001} \\ 
				rs798766 & 0.049 (0.025) & 0.060 \\ 
				rs401681 & 0.077 (0.026) & \textbf{0.007} \\ 
				rs884225 & 0.028 (0.025) & 0.253 \\ 
				rs1057868 & -0.028 (0.026) & 0.881 \\ 
				rs17149580 & 0.015 (0.026) & 0.341 \\ 
				rs12666814 & 0.013 (0.025) & 0.341 \\ 
				rs73223045 & 0.016 (0.025) & 0.341 \\ 
				rs41515546 & 0.015 (0.025) & 0.341 \\ 
				rs12673089 & 0.016 (0.025) & 0.341 \\ 
				rs17149628 & 0.016 (0.024) & 0.341 \\ 
				rs17149630 & 0.016  (0.026) & 0.341 \\ 
				rs17149636 & 0.016  (0.025) & 0.341 \\ 
				rs1495741 & 0.073  (0.026) & \textbf{0.008} \\ 
				rs9642880 & 0.092  (0.027) & \textbf{0.002} \\ 
				rs2294008 & 0.103  (0.026) & \textbf{0.001} \\ 
				rs142492877 & 0.014  (0.026) & 0.341 \\ 
				rs907611 & 0.024  (0.025) & 0.299 \\
				rs217727 & -0.002  (0.027) & 0.567 \\ 
				rs9344 & -0.041  (0.026) & 0.939 \\  
				rs4907479 & 0.072  (0.025) & \textbf{0.007} \\ 
				rs17674580 & 0.090  (0.024)& \textbf{0.001} \\ 
				rs1058396 & 0.047  (0.026) & 0.079 \\
				rs8102137 & 0.081  (0.026) & \textbf{0.006} \\  
				rs62185668 & 0.068  (0.025) & \textbf{0.009} \\ 
				rs6104690 & 0.025  (0.026) & 0.299 \\
				rs4813953 & 0.073  (0.025) & \textbf{0.007} \\ 
				rs1014971 & 0.067  (0.028) & \textbf{0.022} \\ 
				
				\hline
			\end{tabular}
			\caption{Replicability analysis of 31 SNPs based on the UKB UBC data: estimated effects, standard errors (in parentheses) based on Bootstrap 2, and BH-adjusted p-values for the pairwise pseudolikelihood with $K_n=100$. Significant effects at the 0.05 threshold are marked in bold. \label{Tab:SNPs100b}}
		\end{table}

		\begin{figure}
			\centering
			\includegraphics[width=115mm]{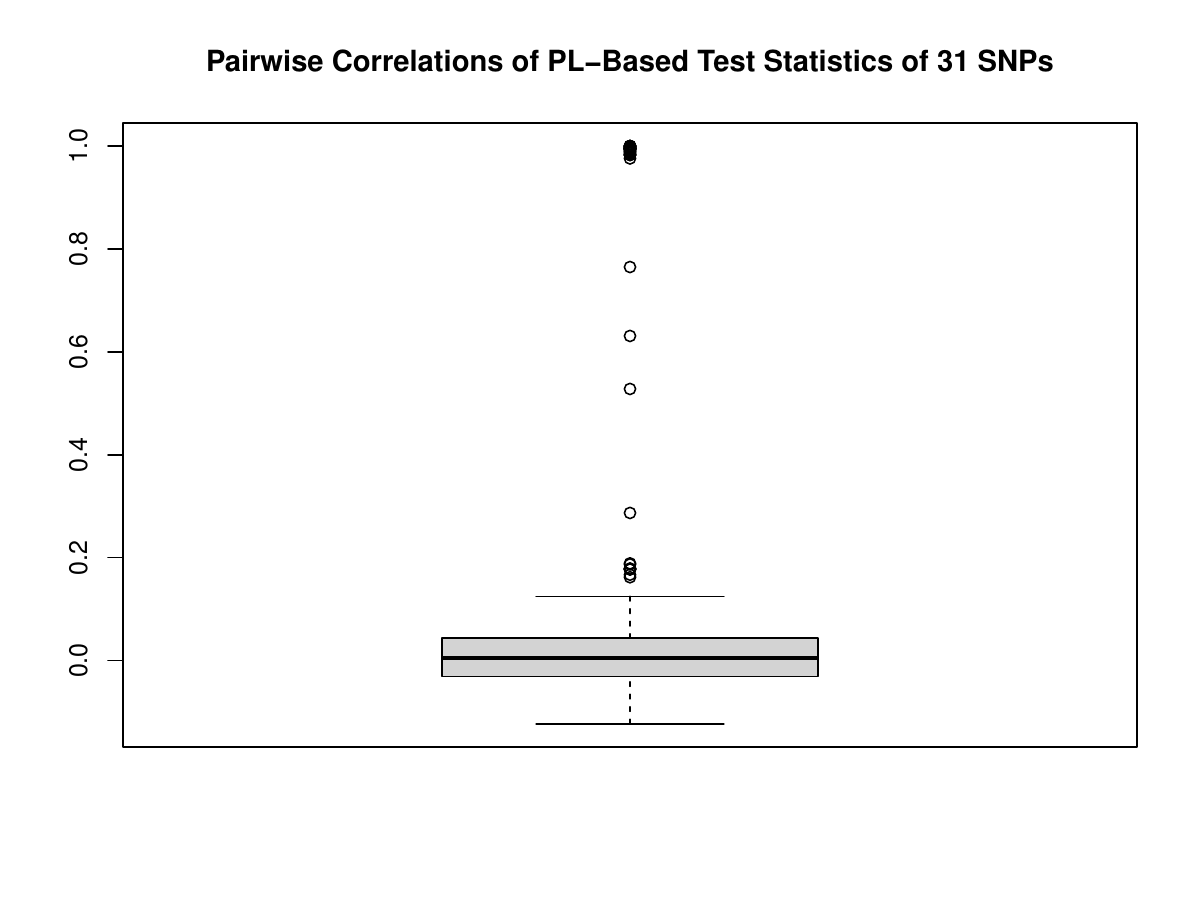}
			\caption{Boxplot of all pairwise correlations among the PL-based test statistics of 31 SNPs. \label{Fig:SNPCorBoxplot}}
		\end{figure}
		
\begin{table}[ht]
	\spacingset{1}
	\centering
	\begin{tabular}{llccccc}
		\hline
		& SNP & est. effect & Boot3-SE & Boot2-SE & adj. p-value(Boot3) & adj. p-value(Boot2) \\ 
		\hline
		1 & rs11892031 & 0.055 & 0.027 & 0.027 & 0.051 & 0.055 \\ 
		2 & rs1052133 & 0.014 & 0.025 & 0.025 & 0.347 & 0.346 \\ 
		3 & rs10936599 & 0.041 & 0.026 & 0.025 & 0.119 & 0.113 \\ 
		4 & rs710521 & 0.096 & 0.026 & 0.026 & \textbf{0.001} & \textbf{0.001} \\ 
		5 & rs798766 & 0.050 & 0.025 & 0.025 & 0.055 & 0.055 \\ 
		6 & rs401681 & 0.075 & 0.026 & 0.026 & \textbf{0.007} & \textbf{0.007} \\ 
		7 & rs884225 & 0.028 & 0.026 & 0.026 & 0.257 & 0.261 \\ 
		8 & rs1057868 & -0.014 & 0.025 & 0.027 & 0.730 & 0.719 \\ 
		9 & rs17149580 & 0.014 & 0.026 & 0.026 & 0.347 & 0.346 \\ 
		10 & rs12666814 & 0.012 & 0.026 & 0.025 & 0.354 & 0.350 \\ 
		11 & rs73223045 & 0.015 & 0.026 & 0.025 & 0.347 & 0.346 \\ 
		12 & rs41515546 & 0.013 & 0.026 & 0.025 & 0.354 & 0.350 \\ 
		13 & rs12673089 & 0.014 & 0.026 & 0.025 & 0.347 & 0.346 \\ 
		14 & rs17149628 & 0.016 & 0.026 & 0.025 & 0.347 & 0.346 \\ 
		15 & rs17149630 & 0.016 & 0.026 & 0.026 & 0.347 & 0.346 \\ 
		16 & rs17149636 & 0.016 & 0.026 & 0.026 & 0.347 & 0.346 \\ 
		17 & rs1495741 & 0.073 & 0.026 & 0.025 & \textbf{0.007} & \textbf{0.007} \\ 
		18 & rs9642880 & 0.092 & 0.026 & 0.026 & \textbf{0.001} & \textbf{0.002} \\ 
		19 & rs2294008 & 0.105 & 0.026 & 0.027 & \textbf{0.001} & \textbf{0.001} \\ 
		20 & rs142492877 & 0.015 & 0.026 & 0.026 & 0.347 & 0.346 \\ 
		21 & rs907611 & 0.021 & 0.025 & 0.026 & 0.347 & 0.346 \\ 
		22 & rs217727 & -0.001 & 0.025 & 0.026 & 0.557 & 0.557 \\ 
		23 & rs9344 & -0.045 & 0.026 & 0.025 & 0.957 & 0.961 \\ 
		24 & rs4907479 & 0.073 & 0.025 & 0.024 & \textbf{0.007} & \textbf{0.007} \\ 
		25 & rs17674580 & 0.092 & 0.025 & 0.025 & \textbf{0.001} & \textbf{0.001} \\ 
		26 & rs1058396 & 0.046 & 0.025 & 0.026 & 0.077 & 0.078 \\ 
		27 & rs8102137 & 0.081 & 0.025 & 0.025 & \textbf{0.003} & \textbf{0.004} \\ 
		28 & rs62185668 & 0.066 & 0.025 & 0.025 & \textbf{0.013} & \textbf{0.014} \\ 
		29 & rs6104690 & 0.029 & 0.025 & 0.026 & 0.246 & 0.260 \\ 
		30 & rs4813953 & 0.075 & 0.025 & 0.026 & \textbf{0.007} & \textbf{0.007} \\ 
		31 & rs1014971 & 0.066 & 0.026 & 0.027 & \textbf{0.018} & \textbf{0.022} \\ 
		\hline
	\end{tabular}
\caption{Replicability analysis of 31 SNPs based on the UKB UBC data: Replicability analysis of 31 SNPs based on the UKB UBC data: estimated effects, standard errors, and BH-adjusted p-values for the proposed pairwise pseudolikelihood with $K_n=150$. Significant effects at the 0.05 threshold are marked in bold. \label{Tab:SNPs150}}
\end{table}
		
	\end{appendices}
	
\FloatBarrier
	\spacingset{2}
	\bibliographystyle{chicago}
	\bibliography{library}
	
\end{document}